\documentclass[11pt,a4paper]{article}
\usepackage[a4paper,hmargin=2.8cm,vmargin=3cm]{geometry}
\usepackage[utf8x]{inputenc}
\usepackage{amsmath,amsthm,amsfonts,amssymb}

\usepackage[bookmarksnumbered=true]{hyperref}

\let\oldmarginpar\marginpar
\renewcommand\marginpar[1]{\-\oldmarginpar[\raggedleft\footnotesize #1]%
  {\raggedright\footnotesize #1}}

\newtheorem{thm}{Theorem}[section]  

\newtheorem{prp}[thm]{Proposition}
\newtheorem{lem}[thm]{Lemma}

\newcommand{\fock}{\mathcal{F}}		
\newcommand{\di}{{d}}		
\newcommand{\Ncal}{\mathcal{N}}		
\newcommand{\Vcal}{\mathcal{V}}		
\newcommand{\hc}{\mbox{h.c.}}		
\newcommand{\scal}[2]{\big<#1,#2\big>} 
\newcommand{\cc}[1]{\overline{#1}}	
\newcommand{\Rbb}{\mathbb{R}}		
\renewcommand{\Re}{\operatorname{Re}\,} 	
\renewcommand{\Im}{\operatorname{Im}\,} 	

\newcommand{\norm}[1]{\lVert#1\rVert}	
\newcommand{\Tr}{\operatorname{tr}}	
\newcommand{\tr}{\operatorname{tr}}
\newcommand{\HS}{_{\textrm{HS}}}
\newcommand{\sgn}{\operatorname{sgn}}

\newcommand{\eps}{\varepsilon}

\newcommand{\bR}{{\mathbb R}}

\newcommand{\bx}{{\bf{x}}}

\newcommand{\bZ}{{\mathbb{Z}}}

\newcommand{\be}{\begin{equation}}
\newcommand{\ee}{\end{equation}}

\newcommand{\cF}{{\cal F}}
\newcommand{\cH}{{\cal H}}
\newcommand{\cL}{{\cal L}}
\newcommand{\cE}{{\cal E}}
\newcommand{\cN}{{\cal N}}
\newcommand{\cU}{{\cal U}}
\newcommand{\cW}{{\cal W}}

\newcommand{\im}{{\text{Im} }}

\newcommand{\bN}{{\mathbb N}}

\newcommand{\bC}{{\mathbb C}}

\newcommand{\ph}{{\varphi}}

\newcommand{\wt}{\widetilde}

\def\g{{\gamma}}
\def\r{{\rho}}

\def\o{{\omega}}
\def\n{{\nu}}

\def\G{{\Gamma}}

\def \bea{\begin{eqnarray}}
\def \eea{\end{eqnarray}}
\def\nn{{\nonumber}}
\def\media#1{{\langle#1\rangle}}

\title{Mean-field Evolution of Fermionic Systems} 

\author{Niels Benedikter, Marcello Porta\thanks{Supported by ERC Grant MAQD 240518 } \, and Benjamin Schlein\thanks{Partially supported by ERC Grant MAQD 240518}\\ \\ Institute of Applied Mathematics, University of Bonn\\ Endenicher Allee 60, 53115 Bonn, Germany}

\begin{document}
\maketitle

\begin{abstract}
The mean field limit for systems of many fermions is naturally coupled with a semiclassical limit. This makes the analysis of the mean field regime much more involved, compared with bosonic systems. 
In this paper, we study the dynamics of initial data close to a Slater determinant, whose reduced one-particle density is an orthogonal projection $\omega_N$ with the appropriate semiclassical structure. Assuming some regularity of the interaction potential, we show that the evolution of such an initial data remains close to a Slater determinant, with reduced one-particle density given by the solution of the Hartree-Fock equation with initial data $\omega_N$. Our result holds for all (semiclassical) times, and 
gives effective bounds on the rate of the convergence towards the Hartree-Fock dynamics.
\end{abstract}
\section{Introduction}
\setcounter{equation}{0}

In the last years, important progress has been achieved in the mathematical understanding of the many body dynamics of bosonic systems in the mean field limit. A system of $N$ bosons in the mean field regime can be described by the Hamiltonian 
\begin{equation}\label{eq:ham-bos} H^{\text{bose}}_N = \sum_{j=1}^N -\Delta_j + \frac{1}{N} \sum_{i<j}^N V (x_i -x_j) \end{equation}
acting on the Hilbert space $L^2_s (\bR^{3N})$, the subspace of $L^2 (\bR^{3N})$ consisting of functions symmetric with respect to an arbitrary permutation of the $N$ particles. 

\medskip

Typical initial data are prepared by confining the system in a volume of order one (for example, restricting the Hamiltonian (\ref{eq:ham-bos}) to $L^2 (\Lambda^N)$, for a cube $\Lambda \subset \bR^3$ of volume one, or by adding a trapping external potential), and letting it relax to the ground state (which is experimentally achieved by cooling it down to very low temperatures). For large $N$, these initial data are approximately factorized, having the form $\psi_N (x_1,\ldots x_N) \simeq \prod_{j=1}^N \ph (x_j)$ for an appropriate one-particle wave function $\ph \in L^2 (\bR^3)$ (obtained as minimizer of the corresponding Hartree energy functional). For such data, the coupling constant $1/N$ in front of the interaction guarantees that both parts of the Hamiltonian are of the order $N$ and that the total potential acting on every particle, given by the sum of a large number (order $N$) of weak contributions (order $1/N$), can be effectively approximated by an average, mean field, interaction. As a 
consequence, factorization is approximately preserved in time, and the many body time evolution generated by (\ref{eq:ham-bos}) can be effectively described in terms of the nonlinear one-particle Hartree dynamics. 

\medskip

To be more precise, for $\psi_N \in L^2_s (\bR^{3N})$, we consider the solution $\psi_{N,t} = e^{-iH_N t} \psi_N$ of the $N$-particle Schr\"odinger equation
\[ i\partial_t \psi_{N,t} = H_N \psi_{N,t}\,. \]
We define the reduced one-particle density $\gamma^{(1)}_{N,t}$ associated with $\psi_{N,t}$ as the non-negative trace class operator on $L^2 (\bR^3)$ with integral kernel
\[ \gamma_{N,t}^{(1)} (x;y) = N \int dx_2 \dots dx_N \, \psi_{N,t} (x, x_2 , \dots , x_N) \overline{\psi}_{N,t} (y, x_2 , \dots , x_N) \,.\]
We normalize the reduced density $\gamma_{N,t}^{(1)}$ so that $\tr \, \gamma_{N,t}^{(1)} = N$. Under suitable assumptions on the potential $V$, it is possible to show that complete condensation is preserved by the time evolution, meaning that
\begin{equation}\label{eq:conv-bos} \frac{1}{N} \, \gamma_{N,t}^{(1)} \to |\ph_t \rangle \langle \ph_t| \quad \text{ as $N \to \infty$,} \end{equation}
for all $t \in \bR$, assuming this to hold at time $t=0$. Here $\ph_t$ is the solution of the nonlinear Hartree equation
\[ i \partial_t \ph_t = (-\Delta + V_{\text{ext}}) \ph_t + (V * |\ph_t|^2) \ph_t \]
with the initial data $\ph_0$ giving the condensate wave function at time $t=0$.
 
 \medskip

The first rigorous proof of this result was obtained in \cite{Sp2}, for bounded interaction potentials. 
In \cite{EY}, the method of \cite{Sp2} was extended to prove the convergence (\ref{eq:conv-bos}) for particles interacting through a Coulomb potential $V(x) = \pm 1/ |x|$. More recently, these techniques have been applied in \cite{ESY1,ESY2,ESY3,ESY4} to systems of bosons in the so called Gross-Pitaevskii limit; in this case, the interaction is scaled so that its range and its scattering length are both of the order $1/N$. As a consequence, the limiting Hartree equation has a local nonlinearity (nonlinear Schr\"odinger equations with local nonlinearity in one and two dimensions have been derived from many body quantum mechanics in appropriate mean field limits in \cite{AGT,KSS}). In contrast with the mean field regime, in the three-dimensional Gross-Pitaevskii limit collisions among particles are rare and strong and the solution of the Schr\"odinger equation develops a singular short scale correlation structure. 

\medskip

Inspired by ideas from \cite{He,GV}, a different approach was developed in \cite{RS} to control the rate of the convergence towards the Hartree equation in the mean field limit of many body quantum dynamics. For a certain class of initial data, and for interaction potentials allowing Coulomb singularities, it was shown in \cite{RS,CLS} that 
\[ \tr \, \left| \gamma_{N,t}^{(1)} - N |\ph_t \rangle \langle \ph_t| \right| \leq C e^{K |t|} \]

\medskip

Among many other important contributions to the mathematical understanding of the many body dynamics for bosonic systems in the mean field regime, let us also recall \cite{FKS}, \cite{AN}, \cite{P,KP}, and the series of papers \cite{GMM1,GMM2,GM}. In \cite{FKS}, the mean field limit for particles interacting through a Coulomb potential was revisited and interpreted as a Egorov type theorem. For regular interactions, the convergence towards the Hartree dynamics was stated in \cite{AN} as propagation of Wigner measures. In \cite{KP,P}, a different approach to obtain control of the rate of the convergence towards the Hartree evolution was proposed; its main advantage compared with other techniques is the fact that it can be extended to potentials with more severe singularities. In \cite{GMM1,GMM2,GM}, on the other hand, it was shown how a more precise approximation of the many body dynamics can be obtained by considering also the next order corrections to the Hartree dynamics. Finally, we remark that also the study 
of the spectral properties of bosonic mean field Hamiltonians received a lot of attention in the last few years. A first proof of the emergence of Bogoliubov excitation spectrum has been found in \cite{S} (for systems of bosons in a box) and in \cite{GS} (in the presence of an external potential). A new and more general approach to the analysis of the excitations spectra of bosonic mean field systems was then obtained in \cite{LNSS}. 

\medskip

In contrast with this long list of results concerning the mean field dynamics of bosons, much less is known for the mean field limit of fermionic systems. It turns out that, in the fermionic case, the mean field regime is naturally linked with a semiclassical limit. Consider a system of fermions initially confined in a volume of order one. Because of the Pauli principle, the kinetic energy of a system of $N$ fermions confined in a volume of order one is at least of the order $N^{5/3}$, much larger than in the bosonic case. In order for the potential energy to be comparable with the kinetic energy, the coupling constant in front of the interaction should be of the order $N^{-1/3}$ (in contrast with the coupling constant of order $N^{-1}$ in the bosonic case). Because of the large kinetic energy, particles move very fast. The average kinetic energy per particle is of the order $N^{2/3}$, and hence the average velocity of the particles is of the order $N^{1/3}$. This means that one can only expect to follow the 
evolution of fermionic systems in the mean field regime for times of the order $N^{-1/3}$. As a consequence, the relevant time-dependent Schr\"odinger equation has the form
\begin{equation}\label{eq:schrf1} i N^{1/3} \partial_t \psi_{N,t} = \left[ \sum_{j=1}^N -\Delta_j + \frac{1}{N^{1/3}} \sum_{i<j}^N V(x_i -x_j) \right] \psi_{N,t}\,. \end{equation}
With this convention, we are interested in times $t$ of order one (so that $\tau = N^{-1/3} t$ is small, of order $N^{-1/3}$). Let $\hbar = N^{-1/3}$. Multiplying (\ref{eq:schrf1}) by $\hbar^2$, we obtain the Schr\"odinger equation 
\begin{equation}\label{eq:schf2}
i \hbar \partial_t \psi_{N,t} = \left[ - \sum_{j=1}^N  \hbar^2 \Delta_j + \frac{1}{N} \sum_{i<j}^N V(x_i -x_j) \right] \psi_{N,t} \,.
\end{equation}
The mean field scaling, characterized by the $1/N$ coupling constant in front of the potential energy, is therefore combined, for fermionic systems, with a semiclassical limit characterized by a small $\hbar = N^{-1/3} \ll 1$. 

\medskip

We observe that a different mean field regime, characterized by $\hbar =1$ in (\ref{eq:schf2}), has been considered in \cite{FK,BGGM}. This alternative scaling may describe physically interesting situations if the particles occupy a large volume (so that the kinetic energy per particle is of order one) and if the interaction has a long range (to make sure that also the potential energy per particle is of order one). In this paper, we will be interested in the evolution of initial data describing $N$ fermions in a volume of order one; correspondingly, we will only consider the scaling appearing in (\ref{eq:schf2}), with $\hbar = N^{-1/3}$. 

\medskip

Similarly to the bosonic case, typical initial data can be prepared by confining the $N$ fermions in a volume of order one and by cooling them down to very low temperatures. In other words, interesting initial data for (\ref{eq:schf2}) are ground states of Hamilton operators of the form
\begin{equation}\label{eq:HNtrap} H_N^{\text{trap}} = \sum_{j=1}^N \left(-\hbar^2 \Delta_{x_j} + V_{\text{ext}} (x_j) \right) + \frac{1}{N} \sum_{i<j}^N V(x_i -x_j) \end{equation} where $V_{\text{ext}}$ is an external trapping potential, confining the $N$ particles in a volume of order one. Such initial data are well approximated by Slater determinants
 \[ \psi_{\text{slater}} (x_1, \dots , x_N) = \frac{1}{\sqrt{N!}} \sum_{\pi \in S_N} \sgn(\pi) f_1 (x_{\pi(1)}) f_2 (x_{\pi(2)}) \dots f_N (x_{\pi(N)}) \] with a family of $N$ orthonormal orbitals $\{ f_j \}_{j=1}^N$ in $L^2 (\bR^3)$ ($\sgn(\pi)$ denotes the sign of the permutation $\pi \in S_N$). Slater determinants are quasi-free states, i.e. they are completely characterized by their one-particle reduced density, given by the orthogonal projection 
\[ \omega = \sum_{j=1}^N |f_j \rangle \langle f_j| \, . \]
In fact, a simple computation shows that $\langle \psi_\text{slater} , H_N^\text{trap} \psi_{\text{slater}} \rangle$ is given by the Hartree-Fock energy 
\begin{equation}\label{eq:HF-en} \begin{split} \cE_{\text{HF}} (\omega) = \tr \, \left(- \hbar^2 \Delta  + V_{\text{ext}} \right) \omega 
& + \frac{1}{N} \int dx dy V(x-y) \omega (x;x) \omega (y;y) \\ &- \frac{1}{N} \int dx dy V(x-y) |\omega (x;y)|^2. \end{split} \end{equation}
In analogy with the convergence (\ref{eq:conv-bos}) observed in the bosonic setting, we expect that the evolution determined by the Schr\"odinger equation \eqref{eq:schf2} of an initial Slater determinant approximating the ground state of \eqref{eq:HNtrap} remains close to a Slater determinant, with an evolved reduced one-particle density, given by the solution of the time-dependent Hartree-Fock equation
\begin{equation}\label{eq:HF}
i\hbar \partial_t  \omega_t = \left[ - \hbar^2 \Delta + (V * \rho_t) - X_t , \omega_t \right] 
\end{equation} 
associated with the energy (\ref{eq:HF-en}). Here $\rho_t (x) = N^{-1} \omega_t (x;x)$ is the normalized density associated with the one-particle density $\omega_t$, while $X_t$ is the exchange operator, having the kernel 
\[ X_t (x;y) = N^{-1} V(x-y) \omega_t (x;y)\, .  \]

\medskip

Of course, we cannot expect this last statement to be correct for any initial state close to a Slater determinant. We expect minimizers of the Hartree-Fock energy functional (\ref{eq:HF-en}) to be  characterized by a semiclassical structure which is essential to understand its evolution. In fact, as we will argue next, we expect the kernel $\omega (x,y)$ of reduced density minimizing (or approximately minimizing) the functional (\ref{eq:HF-en}) to be concentrated close to the diagonal and to decay at distances $|x-y| \gg \hbar$. To understand the emergence of this semiclassical structure, and to find good characterizations, let us consider a system of $N$ free fermions moving in a box of volume one, for example with periodic boundary conditions. The ground state of the system is given by the Slater determinant constructed with the $N$ plane waves $f_p (x) = e^{i p x}$ with $p \in (2\pi) \bZ^3$ and $|p| \leq c N^{1/3}$, for a suitable constant $c$ (guaranteeing that the total number of orbitals equals exactly $N$). The corresponding one-particle reduced density has the kernel 
\[ \omega (x;y) = \sum_{|p| \leq c N^{1/3}} e^{i p \cdot (x-y)} \]
where the sum extends over all $p \in (2\pi) \bZ^3$ with $|p| \leq c N^{1/3}$. Letting $q = \hbar p$ (with $\hbar = N^{-1/3}$), we can write
\begin{equation} \omega (x;y) = \sum_{|q| \leq c} e^{iq (x-y)/\hbar} \simeq \frac{1}{\hbar^3} \int_{|q| \leq c} dq \, e^{i q \cdot (x-y)/\hbar} = \frac{1}{\hbar^3} \ph \left( \frac{x-y}{\hbar} \right) \end{equation}
with 
\begin{equation}\label{eq:ph} \ph (\xi) = \frac{4\pi}{|\xi|^2} \left( \frac{\sin (c|\xi|)}{|\xi|} - c \cos (c|\xi|) \right), \quad \xi \in \Rbb^3. \end{equation} 
Hence, at fixed $N$ and $\hbar$, $\omega (x;y)$ decays to zero for $|x-y| \gg \hbar$. Moreover, the fact that $\omega$ depends only on the difference $x-y$ (for $x,y$ in the box) implies that the density $\omega (x;x)$ is constant inside the box (and zero outside). This is of course a consequence of the fact that we are considering a system with external potential vanishing inside the box, and being infinity outside of it. More generically, if particles are trapped by a regular potential $V_{\text{ext}}$ with $V_{\text{ext}} (x) \to \infty$ for $|x| \to \infty$, we expect the resulting reduced one-particle density to have, approximately, the form   
\begin{equation}\label{eq:gamma-sc} \omega (x;y) \simeq \frac{1}{\hbar^3} \ph \left( \frac{x-y}{\hbar} \right) \chi \left( \frac{x+y}{2} \right) \end{equation}
for appropriate functions $\ph$ and $\chi$, or to be linear combinations of such kernels. While $\chi$ determines the density of the particles in space (because $\ph (0) = 1$, to ensure that $\tr \omega = N$), $\ph$ fixes the momentum distribution. 

\medskip

Next we look for suitable bounds, characterizing Slater determinants like (\ref{eq:gamma-sc}) which have the correct semiclassical structure. To this end, we observe that, if we differentiate the r.h.s. of (\ref{eq:gamma-sc}) with respect to $x$ or $y$, a factor $\hbar^{-1}$ will emerge from the derivative of $\ph$ (this produces a kinetic energy of order $N^{5/3}$, as expected). However, if we take the commutator $[\nabla, \omega]$, its kernel will be given by 
\begin{equation}\label{eq:comm-ker} [ \nabla , \omega ] (x;y) = \left( \nabla_x + \nabla_y \right) \omega (x;y) = \frac{1}{\hbar^3} \ph \left( \frac{x-y}{\hbar} \right) \nabla \chi \left( \frac{x+y}{2} \right) \,. \end{equation}
In this case the derivative only hits the density profile $\chi$; it does not affect $\ph$, and therefore it remains of order one (of course, in the example with plane waves in a box, there is the additional problem that $\chi$ is the characteristic function of the box, and therefore that it is not differentiable; this is however a consequence of the pathological choice of the external potential, which is infinity outside the box). We express the fact that the derivative in (\ref{eq:comm-ker}) does not produce additional $\hbar^{-1}$ factors through the bound
\begin{equation}\label{eq:tr-nabla} \tr \, \left| \left[\nabla , \omega \right] \right| \leq C N. \end{equation}
Similarly, the fact that $\omega (x;y)$ decays to zero as $|x-y| \gg \hbar$, suggests that the commutator $[x, \omega ]$, whose kernel is given by
\begin{equation}\label{eq:commx} [ x , \omega ] (x;y) = (x-y) \omega (x;y), \end{equation} 
is smaller than $\omega$, by order $\hbar$. In fact, one has to be a bit careful here. Going back to the plane wave example, we observe that the function $\ph$ computed in (\ref{eq:ph}) does not decay particularly fast at infinity. For this reason, it is not immediately clear that one can extract an $\hbar$ factor from the difference $(x-y)$ on the r.h.s. of (\ref{eq:commx}). Keeping in mind the plane-wave example, let us compute the commutator of the reduced density $\omega$ with the multiplication operator $e^{i r \cdot x}$, for a fixed $r \in (2\pi) \bZ^3$. We find
\[ \left[ e^{ir \cdot x} , \omega \right] = \sum_{|p| \leq c N^{1/3}} \left[ |e^{i (r+p) \cdot x} \rangle \langle e^{ip\cdot x} | - |e^{i p\cdot x} \rangle \langle e^{i (p-r) \cdot x} | \right]. \]
A straightforward computation shows that 
\begin{equation}\label{eq:preHS} \left| \left[ e^{ir \cdot x} , \omega \right]  \right|^2 = \sum_{p \in I_r} |e^{i p \cdot x} \rangle \langle e^{i p \cdot x}| \end{equation}
where \[ \begin{split} I_r = &\, (2\pi) \bZ^3 \cap  \left\{ p \in \bR^3 : |p-r| \leq c N^{1/3},  |p| \geq c N^{1/3} \; \text{or } \,  |p-r| \geq c N^{1/3} , |p| \leq c N^{1/3} \right\}. \end{split}\]
It follows that $|[ e^{ir \cdot x} , \omega]| = |[ e^{ir \cdot x} , \omega]| ^2$ is a projection, and therefore that 
\begin{equation}\label{eq:tr-exp} 
\tr\; \left|\left[ e^{ir \cdot x} , \omega \right] \right| \leq C N \hbar |r|. \end{equation}
Hence, the trace norm of the commutator is smaller, by a factor $\hbar$, compared with the norm of the operators $e^{ip\cdot x} \omega$ and $\omega e^{ip\cdot x}$. The fact that the kernel $\omega (x;y)$ is supported close to the diagonal allows us to extract an additional $\hbar$-factor from the trace norm of the commutator $[e^{ip\cdot x} ,\omega]$. Notice, however, that if we considered the Hilbert-Schmidt norm of $[e^{ip\cdot x} , \omega]$, we would find from \eqref{eq:preHS} that
\[\norm{\left[ e^{ir \cdot x} , \omega \right]}_{HS} =  \left(\tr  \left|\left[ e^{ir \cdot x} , \omega \right] \right|^2\right)^{1/2} \leq (C N \hbar |r|)^{1/2}. \]
In other words, the Hilbert Schmidt norm of the commutator $[e^{ip\cdot x} , \omega]$ is only smaller than the Hilbert-Schmidt norm of the two operators $e^{ip \cdot x}\omega$ and $\omega e^{ip \cdot x}$ by a factor $\hbar^{1/2}$. This is consistent with the fact that, in (\ref{eq:gamma-sc}), the function $\ph$ does not decay fast at infinity (which follows from the fact that $\omega$ is a projection corresponding to a characteristic function in momentum space). 

\medskip

So far, we proved the bounds (\ref{eq:tr-nabla}) and (\ref{eq:tr-exp}) for minimizers of systems of confined non-interacting electrons. What happens now if we turn on a mean-field interaction? Can we still expect the minimizer of the Hamiltonian (\ref{eq:HNtrap}) to satisfy (\ref{eq:tr-nabla}) and (\ref{eq:tr-exp})? We claim that the answer to this question is affirmative, and we propose a heuristic explanation\footnote{We would like to thank Rupert Frank for pointing out this argument to us.}. Semiclassical analysis suggests that the reduced density of the minimizer of (\ref{eq:HNtrap}) can be approximated by the Weyl quantization $\omega = \text{Op}^w_M$ of the phase space density $M(p,x) = \chi (|p| \leq (6 \pi^2 \rho (x))^{1/3})$, where $\rho$ is the minimizer of the Thomas-Fermi type functional
\[ \eps_{\text{TF}} (\rho) = \frac{3}{5} (3\pi^2)^{2/3} \int dx \rho^{5/3} (x) + \int dx \, V_{\text{ext}} (x) \rho (x) + \frac{1}{2} \int dx dy V(x-y) \rho (x) \rho (y) \]
over all non-negative densities $\rho \in L^1 (\bR^3) \cap L^{5/3} (\bR^3)$ normalized so that $\| \rho \| = 1$. Here, the Weyl quantization $\text{Op}^w_M$ of $M$ is defined by the kernel
\[ \omega (x,y) = \text{Op}^w_M (x,y) = \frac{1}{(2\pi \hbar)^3} \int dp \, M \left( p, \frac{x+y}{2} \right) \, e^{i p \cdot \frac{x-y}{\hbar}} \] 
It turns out that the commutators of $\omega$ with the position operator $x$ and with the momentum operator $\nabla$ are again Weyl quantizations. In fact, a straightforward computation shows that 
\[ \begin{split}  [x, \omega] &= -i \hbar \text{Op}^w_{\nabla_p M}, \quad [ \nabla , \omega] = \text{Op}^w_{\nabla_q M}. \end{split} \] Hence, semiclassical analysis predicts that
\[ \tr \, |[x,\omega]| \simeq \frac{\hbar}{(2\pi \hbar)^3} \int dp dq |\nabla_p M (p,q)| = CN \hbar \,  \int  \rho^{2/3} (q) dq  \]
and that 
\[ \tr \, |[\nabla , \omega]| \simeq \frac{1}{(2\pi \hbar)^3} \int dp dq |\nabla_q M (p,q)| = N \int \, |\nabla \rho (q)| \,dq \]
Under general assumptions on $V_\text{ext}$ and $V$, we can expect the integrals on the r.h.s. of the last two equations to be finite, and therefore, we can expect the bounds (\ref{eq:tr-nabla}) and (\ref{eq:tr-exp}) to hold true ((\ref{eq:tr-exp}) easily follows from the estimate $\tr |[x,\omega ]| \leq C N \hbar$). Although one could probably turn the heuristic argument that we just presented into a rigorous proof, we do not pursue this question in the present work. Instead, we will just assume our initial data to satisfy (\ref{eq:tr-nabla}) and (\ref{eq:tr-exp}). We consider these bounds as an expression of the semiclassical structure that emerges naturally when one considers states with energy close to the ground state of a trapped Hamiltonian of the form (\ref{eq:HNtrap}). 

\medskip

For initial data $\psi_N$ close to Slater determinants and having the correct semiclassical structure characterized by (\ref{eq:tr-nabla}) and (\ref{eq:tr-exp}), we consider the time evolution 
$\psi_{N,t} = e^{-i t H_N / \hbar} \psi_N$, generated by the Hamiltonian 
\begin{equation}\label{eq:HN-fer} H_N = - \sum_{j=1}^N \hbar^2 \Delta_j + \frac{1}{N} \sum_{i<j}^N V (x_i -x_j) \end{equation}
and we denote by $\gamma_{N,t}^{(1)}$ the one-particle reduced density associated with $\psi_{N,t}$. Our main result, Theorem \ref{thm:main}, shows that, under suitable assumptions on the potential $V$, 
there exist constants $K,c_1,c_2 > 0$ such that
\begin{equation}\label{eq:diff-hs} \| \gamma^{(1)}_{N,t} - \omega_{N,t} \|_{\text{HS}} \leq K \exp (c_1 \exp (c_2 |t|)) 
\end{equation}
and 
\begin{equation}\label{eq:conv-fer} \tr \; \left| \gamma^{(1)}_{N,t} - \omega_{N,t} \right| \leq K N^{1/6} \exp (c_1 \exp (c_2 |t|)) 
\end{equation}
where $\omega_{N,t}$ denotes the solution of the time-dependent Hartree-Fock equation (\ref{eq:HF})
with the initial data $\omega_{N,t=0} = \gamma_{N,0}^{(1)}$. The bounds (\ref{eq:diff-hs}) and (\ref{eq:conv-fer}) show that the difference $\gamma_{N,t}^{(1)} - \omega_{N,t}$ is much smaller (both in the Hilbert-Schmidt norm and in the trace-class norm) than $\gamma_{N,t}^{(1)}$ and $\omega_{N,t}$ (recall that $\| \omega_{N,t}^{(1)} \|_{\text{HS}}, \| \gamma_{N,t}^{(1)} \|_{\text{HS}} \simeq N^{1/2}$ while $\tr \omega_{N,t} , \tr \gamma^{(1)}_{N,t} \simeq N$). 

\medskip

It turns out that the contribution of the exchange term is small compared to the other terms in the Hartree-Fock equation (\ref{eq:HF}); in fact, for the class of regular potential that we will consider in this paper, it is of the relative size $1/N$. As a consequence, the bounds (\ref{eq:diff-hs}), (\ref{eq:conv-fer}) and also all other bounds that we prove in Theorem \ref{thm:main} for the difference between $\gamma_{N,t}^{(1)}$ and the solution of the Hartree-Fock equation remain true if we replace the solution of the Hartree-Fock equation $\omega_{N,t}$ by the solution $\wt{\omega}_{N,t}$ of the Hartree equation 
\begin{equation}\label{eq:hartree}
i \hbar \partial_t \wt{\omega}_{N,t} = \left[ -\hbar^2 \Delta + (V * \wt{\rho}_t) , \wt{\omega}_{N,t} \right] \end{equation}
with the same initial data $\wt{\omega}_{N,t=0} = \gamma_{N,0}^{(1)}$ (here $\wt{\rho}_t (x) = N^{-1} \wt{\omega} (x;x)$ is the normalized density associated to $\wt{\omega}_{N,t}$). For more details, see the last remark after Theorem \ref{thm:main} and Proposition~\ref{prop:hartree} in Appendix \ref{app:hartree}.

\medskip

Observe that both the Hartree-Fock equation (\ref{eq:HF}) and the Hartree equation (\ref{eq:hartree}) still depend on $N$, through Planck's constant $\hbar = N^{-1/3}$. In the semiclassical limit $\hbar \to 0$, the Hartree (and the Hartree-Fock) dynamics can be approximated by the solution of the Vlasov equation. We define the Wigner transform $W_{N,t}$ associated with the solution $\omega_{N,t}$ of the Hartree-Fock equation by 
\[ W_{N,t} (x,p) = \frac{1}{(2\pi)^3} \int dy \, \omega_{N,t} \left(x+ \hbar \frac{y}{2} ; x - \hbar \frac{y}{2} \right) e^{-i p y} \]
It is a well-known fact that, in the limit $\hbar \to 0$, the Wigner transform $W_{N,t}$ of the solution of the Hartree-Fock equation (\ref{eq:HF}) (or the Wigner transform of the solution $\wt{\omega}_{N,t}$ of the Hartree equation) converges towards the solution of the Vlasov equation
\begin{equation}\label{eq:vlasov} \partial_t W^{\text{vl}}_t (x,p) + p \cdot \nabla_x W^\text{vl}_{t} (x,p) = \nabla_x \left( V * \rho^\text{vl}_{t} \right) (x) \nabla_p W^\text{vl}_{t} (x,p) \, . \end{equation}
where $\rho^\text{vl}_t (x) = \int dp \, W^{\text{vl}} (x,p)$. The difference between the Wigner transform $W_{N,t}$ of $\omega_{N,t}$ and the solution of the Vlasov equation $W^\text{vl}_{t}$ is of the order $\hbar N = N^{2/3}$, and therefore much larger than the difference between the reduced one-particle density $\gamma_{N,t}^{(1)}$ associated with the solution of the many body Schr\"odinger equation and the solution $\omega_{N,t}$ of the Hartree-Fock equation (or the solution $\wt{\omega}_{N,t}$ of the Hartree equation). In other words, the Hartree-Fock approximation (or the Hartree approximation) keeps the quantum structure of the problem and gives a much more precise approximation of the many body evolution compared with the classical Vlasov dynamics. Our result is therefore a dynamical counterpart to \cite{Ba,GSo}, where the Hartree-Fock theory is shown to give a much better approximation to the ground state energy of a system of atoms or molecules as compared with the Thomas-Fermi energy (but, of course, in 
contrast to \cite{Ba,GSo}, our analysis does not apply so far to a Coulomb interaction). 

\medskip

As mentioned above, the literature on the mean field dynamics of fermionic systems is rather limited. As far as we know, the first rigorous results concerning the evolution of fermionic system in the regime we are interested in was obtained in \cite{NS}, where the authors prove that, for real analytic potential, the Wigner transform of the reduced density $\gamma_{N,t}^{(1)}$ associated with the solution of the Schr\"odinger equation converges weakly to the solution of the Vlasov equation (\ref{eq:vlasov}). The regularity assumptions were substantially relaxed in \cite{Sp}. Notice that neither \cite{NS} nor \cite{Sp} give a bound on the rate of the convergence. More recently, in \cite{EESY} the many body evolution is compared with the $N$-dependent Hartree dynamics described by (\ref{eq:hartree}); under the assumption of a real analytic potential, it is shown that, for short semiclassical times, the difference between $\gamma_{N,t}^{(1)}$ and $\wt{\omega}_{N,t}$, when tested against appropriate observables, 
is of the order $N^{-1}$. The results of our paper are comparable with those of \cite{EESY}, but we obtain convergence for arbitrary semiclassical times (for arbitrary $t$ of order one, where $t$ is the time variable appearing in (\ref{eq:schf2})) and under much weaker regularity conditions on the interaction potential. It should also be noted that the mean field limit of fermionic systems with a different scaling (the same scaling used in (\ref{eq:ham-bos}) for bosonic systems) has been considered in \cite{BGGM} (for regular interactions) and in \cite{FK} (for potentials with Coulomb singularity). On the other hand, we remark that a joint mean field and semiclassical limit has been considered, for bosonic systems, in \cite{GMP} and \cite{FGS}.
 
\section{Fock space representation and quasi-free states}\label{secfock}
\setcounter{equation}{0}

The fermionic Fock-space over $L^2 (\bR^3)$ is defined by 
\[ \cF = \bigoplus_{n \geq 0} L^2_a (\bR^{3n}, dx_1 \dots dx_n) \]
where $L^2_a (\bR^{3n})$ is the subspace of $L^2 (\bR^{3n})$ consisting of all functions which are antisymmetric with respect to permutation of the $n$ particles. In other words,
\[ L^2_a (\bR^{3n}) = \{ f \in L^2 (\bR^{3n}) : f( x_{\pi(1)}, \dots , x_{\pi(n)}) = \sgn(\pi) f (x_1 , \dots , x_n) \text{ for all $\pi \in S_n$} \}. \]
Here $\sgn(\pi)$ denotes the sign of the permutation $\pi \in S_n$. 

\medskip

For a one-particle operator $O$, acting on $L^2 (\bR^3)$, we denote its second quantization by $d\Gamma (O)$. This is an operator on $\cF$, defined by 
\[ \left(d\Gamma (O) \psi \right)^{(n)} = \sum_{j=1}^n O^{(j)} \psi^{(n)} \]
where $O^{(j)}$ denotes the operator $O$ acting only on the $j$-th particle (i.e. $O^{(j)} = 1^{\otimes (j-1)} \otimes O \otimes 1^{\otimes (n-j)}$). An important example is the number of particles operator, defined by $\cN = d\Gamma (1)$. 

\medskip

On $\cF$, it is useful to introduce creation and annihilation operators. For $f \in L^2 (\bR^3)$, we define 
\[ \begin{split} (a^* (f) \psi)^{(n)} (x_1, \dots , x_n) &= \frac{1}{\sqrt{n}} \sum_{j=1}^n  (-1)^j f(x_j) \psi^{(n-1)} (x_1, \dots , x_{j-1}, x_{j+1}, \dots , x_n), \\ 
(a(f) \psi)^{(n)} (x_1 , \dots , x_n) &= \sqrt{n+1} \int dx \overline{f} (x) \psi^{(n+1)} (x, x_1, \dots , x_n).
\end{split}\]
Observe here that creation operators are linear while annihilation operators are antilinear in their argument. They satisfy canonical anticommutation relations 
\begin{equation}\label{eq:CAR} \{ a (f) , a^* (g) \} = \langle f , g \rangle, \quad \{ a (f) , a(g) \} = \{ a^* (f) , a^* (g) \} = 0 \end{equation}
for all $f,g \in L^2 (\bR^3)$. It is also important to note that, in contrast to the bosonic case, fermionic creation and annihilation operators are bounded. In fact 
\[ \| a (f) \psi \|^2 = \langle a (f) \psi , a(f) \psi \rangle = \langle \psi, a^* (f) a(f) \psi \rangle = \| f \|_2^2 - \langle \psi, a(f) a^* (f) \psi \rangle \leq \| f \|_2^2 \]
and a similar computation for $\| a^* (f) \psi \|^2$ imply that
\begin{equation}\label{eq:bdaa*} 
\| a(f) \| \leq \| f \|_2 \quad \text{ and } \quad \| a^* (f) \| \leq \| f \|_2.
\end{equation}
It is also useful to introduce operator valued distributions $a_x^*$ and $a_x$, which formally create, respectively, annihilate a particle at the point $x \in \bR^3$. They are defined by 
\[ a (f) = \int dx \, \overline{f (x)} \, a_x , \quad a^* (f) = \int dx \, f(x) \, a^*_x \]
for all $f \in L^2 (\bR^3)$. In terms of these operator valued distributions, it is possible to write the second quantization $d\Gamma (O)$ of a one-particle operator $O$ with integral kernel $O (x;y)$ as
\be
d\Gamma (O) = \int dx dy \, O (x;y) a_x^* a_y. \label{defsecqu}
\ee
In particular, the number of particles operator is given by
\[ \cN = \int dx \, a_x^* a_x.\]
Observe that, even for bounded $O$, the second quantized operator $d\Gamma (O)$ does not need to be bounded, simply because the number of particles is not bounded. It turns out, however, that because of the fermionic statistics, $d\Gamma (O)$ is a bounded operator if $O$ is trace-class. This fact together with other useful bounds will be shown in Lemma \ref{lm:bds-2} in Section \ref{sec:gro}.

\medskip

Since we want to study the time evolution of fermionic systems, we need to define a Hamilton operator on the Fock space $\cF$. Inspired by (\ref{eq:HN-fer}), we introduce the operator $\cH_N$, by setting $(\cH_N \psi)^{(n)} = \cH_N^{(n)} \psi^{(n)}$, with 
\[ \cH_N^{(n)} = \sum_{j=1}^n - \hbar^2 \Delta_{x_j} + \frac{1}{N} \sum_{i<j}^n V(x_i -x_j) \]
where, as discussed in the introduction, $\hbar =N^{-1/3}$. Hence, the Hamiltonian $\cH_N$ leaves each sector of the Fock space with a fixed number of particles invariant. On the $N$-particle sector, it agrees with (\ref{eq:HN-fer}). Notice that in the notation $\cH_N$, the index $N$ does not refer here to the number of particles, since $\cH_N$ acts on the whole Fock space. It reminds instead of the coupling constant $1/N$ in front of the potential energy, and of the $N$-dependent Planck constant $\hbar = N^{-1/3}$ in front of the kinetic energy. Of course, in order to recover the mean field regime discussed in the introduction, we will consider later the time evolution of states in $\cF$ having approximately $N$ particles. Observe that, in terms of the operator valued distributions $a_x$ and $a_x^*$, we can express the Hamiltonian $\cH_N$ as
\begin{equation}\label{eq:cHN} \cH_N = \hbar^2 \int dx \nabla_x a_x^* \nabla_x a_x + \frac{1}{2N} \int dx dy V(x-y) a_x^* a_y^* a_y a_x. \end{equation}
Notice that the kinetic energy is just given by the second quantization $d\Gamma (-\hbar^2 \Delta)$.

\medskip

It will also be important to consider linear combinations of creation and annihilation operators. For 
$f,g \in L^2 (\bR^3)$ we set \[ A(f,g) = a(f) + a^* (\bar{g}), \quad \text{and } A^* (f,g) = (A(f,g))^* = a^* (f) + a(\bar{g}).\] 
Observe that
\begin{equation}\label{eq:AA*} A^* (f,g) = A (Jg, Jf) \end{equation}
where we introduced the antilinear operator $J : L^2 (\bR^3) \to L^2 (\bR^3)$ defined by $J f = \overline{f}$. Note that $A^*$ is linear while $A$ is antilinear in its two arguments. In terms of the operators $A, A^*$ the canonical anticommutation relations assume the form
\begin{equation}\label{eq:commA} \begin{split} 
\left\{ A (f_1, g_1) , A^* (f_2, g_2) \right\} &= \{ a (f_1) , a^* (f_2) \} + \{ a^* (\bar{g}_1) , a (\bar{g}_2) \} \\ &= \langle f_1, f_2 \rangle + \langle \bar{g}_2 , \bar{g}_1 \rangle = \langle f_1 , f_2 \rangle + \langle g_1 , g_2 \rangle \\ &= \langle (f_1, g_1) , (f_2, g_2) \rangle_{L^2 \oplus L^2}. \end{split} \end{equation}
Note that $\{ A (f_1, g_1) , A(f_2, g_2) \}$ and $\{ A^* (f_1, g_1) , A^* (f_2 , g_2) \}$ in general do not vanish.

\medskip

We now introduce Bogoliubov transformations (a useful review on this subject can be found, for example, in the lecture notes \cite{Solovej}). A (fermionic) Bogoliubov transformation is a linear map $\nu: L^2 (\bR^3)  \oplus L^2 (\bR^3) \to L^2 (\bR^3) \oplus L^2 (\bR^3)$ with the properties
\begin{equation}\label{eq:car-nu} \left\{ A( \nu (f_1 , g_1)) , A^* (\nu (f_2 , g_2))\right \} = \left\{ A (f_1 , g_1) , A^* (f_2 , g_2) \right\} \end{equation}
for all $f_1, g_1, f_2 , g_2 \in L^2 (\bR^3)$, and
\begin{equation}\label{eq:BB*} A^* (\nu (f,g))  = A (\nu (\bar{g}, \bar{f})) \end{equation}
for all $f,g \in L^2 (\bR^3)$. In other words, a Bogoliubov transformation is a map $\nu :L^2 (\bR^3)  \oplus L^2 (\bR^3) \to L^2 (\bR^3) \oplus L^2 (\bR^3)$ with the property that (\ref{eq:AA*}) and the canonical anticommutation relations (\ref{eq:commA}) continue to hold for the new field operators $B(f,g) := A (\nu (f,g))$. Note that, by (\ref{eq:commA}), condition (\ref{eq:car-nu}) means that every Bogoliubov transformation $\nu$ is unitary. The condition (\ref{eq:BB*}), on the other hand, is equivalent to 
\[ \left( \begin{array}{ll} 0 & J \\ J & 0 \end{array} \right) \nu = \nu \left( \begin{array}{ll} 0 & J \\ J & 0 \end{array} \right). \]
It is then simple to check that a linear map $\nu: L^2 (\bR^3) \oplus L^2 (\bR^3) \to L^2 (\bR^3) \oplus L^2 (\bR^3)$ is a Bogoliubov transformation if and only if it has the form
\begin{equation}\label{eq:Nu} \nu =  \left( \begin{array}{ll} u & \overline{v} \\ v & \overline{u} \end{array} \right) \end{equation}
where $u,v: L^2(\Rbb^3) \to L^2(\Rbb^3)$ are linear maps with $u^* u + v^* v = 1$ and $u^* \overline{v} + v^* \overline{u} = 0$. Here we used the notation $\overline{u} = J u J$, for any linear operator $u : L^2 (\bR^3) \to L^2 (\bR^3)$. Notice that, if $u$ is a linear operator with integral kernel $u(x;y)$, then $\overline{u}$ is again a linear operator, with the integral kernel $\overline{u} (x;y) = \overline{u(x;y)}$ (this explain the notation $\overline{u}$). 

\medskip

We say that a Bogoliubov transformation $\nu$ is implementable on the fermionic Fock space $\cF$ if there exists a unitary operator $R_\nu: \fock \to \fock$ with the property
\begin{equation}\label{eq:Bog} R_\nu^* A(f,g) R_\nu = A (\nu (f,g)) \end{equation}
for all $f,g \in L^2 (\bR^3)$. A Bogoliubov transformation (\ref{eq:Nu}) is implementable if and only if $v$ is a Hilbert-Schmidt operator (Shale-Stinespring condition, see e.\,g.\ \cite[Theorem 9.5]{Solovej} or \cite{Ru}). 

\medskip

Given a Fock space vector $\psi \in \cF$, we define the one-particle reduced density $\gamma_\psi$ associated with $\psi$ as the non-negative operator with the integral kernel 
\[ \gamma_\psi (x;y) = \langle \psi , a_y^* a_x \psi \rangle.\]
Notice that $\gamma_\psi$ is normalized such that $\tr \, \gamma_\psi = \langle \psi , \cN \psi \rangle$. Hence $\gamma_\psi$ is a trace class operator if the expectation of $\cN$ in the state $\psi$ is finite.
In general, if $\psi$ does not have a fixed number of particles, it is also important to track the expectations $\langle \psi, a_y a_x \psi \rangle$ and $\langle \psi, a_x^* a_y^* \psi \rangle$. 
We define therefore the pairing density $\alpha_\psi$ associated with $\psi$ as the one-particle operator with integral kernel  
\[ \alpha_\psi (x;y) = \langle \psi , a_y a_x \psi \rangle. \]
Then we also have $\overline{\alpha_\psi} (x;y) = \langle \psi , a_x^* a_y^* \psi \rangle$. The operators $\gamma_\psi$ and $\alpha_\psi$ can be combined into the generalized one-particle density $\Gamma_\psi : L^2 (\bR^3) \oplus L^2 (\bR^3) \to L^2  (\bR^3) \oplus L^2 (\bR^3)$ defined by
\[ \left\langle (f_1, g_1) , \Gamma_\psi (f_2, g_2) \right\rangle = \left\langle \psi, A^* (f_2, g_2) A (f_1, g_1) \psi \right\rangle. \]
A simple computation shows that $\Gamma_\psi$ can be expressed in terms of $\gamma_\psi$ and $\alpha_\psi$ as 
\begin{equation}\label{eq:Gamma} \Gamma_\psi = \left( \begin{array}{ll} \gamma_\psi & \alpha_\psi \\ -\overline{\alpha}_\psi & 1-\overline{\gamma}_\psi \end{array} \right). \end{equation}
As a consequence of the canonical anticommutation relations, it is simple to check that $0 \leq \Gamma_\psi \leq 1$.

\medskip

Knowledge of the generalized one-particle density $\Gamma_\psi$ allows the computation of the expectation of all observables which are quadratic in creation and annihilation operators. To compute expectations of operators involving more than two creation and annihilation operators, one needs higher order correlation functions, having the form
\begin{equation}\label{eq:hcorr} \left\langle \psi, a_{x_1}^\# \dots a_{x_k}^\# \psi \right\rangle \end{equation}
where each $a^\#$ is either an annihilation or a creation operator. An important class of states on $\cF$ are quasi-free states. A pure quasi-free state is a vector in $\cF$ with the form $\psi = R_\nu \Omega$, where $R_\nu$ is the unitary implementor of an (implementable) Bogoliubov transformation $\nu$. The crucial (and defining) property of quasi-free states is the fact that all higher order correlations functions like (\ref{eq:hcorr}) can be expressed, using Wick's theorem, just in terms of the reduced density $\gamma_\psi$ and the pairing density $\alpha_\psi$ (see, for example, \cite[Theorem 10.2]{Solovej}). In other words, quasi-free states are completely described by their generalized one-particle reduced density $\Gamma_\psi$. If $\nu$ is a Bogoliubov transformation of the form (\ref{eq:Nu}), it is simple to check that the reduced one-particle density associated with $\psi$ has the form
\[ \Gamma_\nu = \left( \begin{array}{ll} v^* v & v^* \overline{u} \\ \overline{u}^* v & \overline{u}^* \overline{u} \end{array} \right). \]
Hence, the reduced density of the quasi-free state associated with the Bogoliubov transformation $\nu$ is $\gamma_\nu = v^* v$, while the pairing density is $\alpha_\nu = v^* \overline{u}$. {F}rom the property of Bogoliubov transformations, we conclude that $\gamma_\nu$ is trace class (because $v$ is a Hilbert-Schmidt operator, for $\nu$ to be implementable) and hence the expectation of the number of particles is always finite for quasi-free states. Moreover, it follows that $\Gamma_\nu^2 = \Gamma_\nu$, i.e. $\Gamma_\nu$ is a projection. Conversely, for every linear projection $\Gamma : L^2 (\bR^3) \oplus L^2 (\bR^3) \to L^2 (\bR^3) \oplus L^2 (\bR^3)$ having the form 
\[ \Gamma =  \left( \begin{array}{ll} \gamma & \alpha \\ -\overline{\alpha} & 1-\overline{\gamma} \end{array} \right)  \]
for a trace class operator $\gamma$, there exists a quasi-free state, i.e. an implementable Bogoliubov transformation $\nu$, such that $\Gamma = \Gamma_\nu$, i.e. $\Gamma$ is the generalized reduced density associated with the Fock space state $R_\nu \Omega$. Restricting the  Hamiltonian (\ref{eq:cHN}) on quasi-free states of the form $R_\nu \Omega$ one obtains the Bardeen-Cooper-Schrieffer (BCS) energy functional. BCS theory plays a very important role in physics. Originally introduced to describe superconductors, it has been later applied to explain the phenomenon of superfluidity observed in dilute gases of fermionic atoms at low temperature. In the last years, there has been a lot of progress in the mathematical understanding of BCS theory; see, for example, \cite{HS2,HHSS,FHSS,HS} for results concerning equilibrium properties and \cite{HLLS,HSc} for results about the time-evolution in BCS theory.   

\medskip

In this paper we will be interested in pure quasi-free states with no pairing, i.e. with $\alpha = 0$. Since $\Gamma$ must be a projection, the assumption $\alpha = 0$ implies that $\gamma$ is a projection.  We require the number of particles to be $N$, i.e. $\tr\, \gamma = N$. We know then that there exists a Bogoliubov transformation $\nu$, such that $\gamma$ is the reduced density of $R_\nu \Omega$. In fact, it is easy to construct such a Bogoliubov transformation. Since we assumed $\gamma$ to be an orthogonal projection with $\tr \, \gamma = N$, there must be an orthonormal system $\{ f_j \}_{j=1}^N$ such that $\gamma = \sum_{j=1}^N |f_j \rangle \langle f_j|$. We define then $v = \sum_{j=1}^N |\bar{f}_j \rangle \langle f_j|$. Then we have $v^* = \overline{v} = \sum_{j=1}^N |f_j \rangle \langle \bar{f}_j|$
and $v^* v = \sum_{j=1}^N |f_j \rangle \langle f_j| = \gamma$.
We also set 
\[ u = u^* = 1-\sum_{j=1}^N |f_j \rangle \langle f_j| = 1- \gamma. \]
Then $u$ is a projection and $u^* u = u^2 = u = 1-\gamma$. Hence $u^* u + v^* v = 1$, and $v^* \overline{u} = 0$. It follows that
\begin{equation}\label{eq:Nu-gamma} \nu =\left( \begin{array}{ll} u & \overline{v}  \\ v & \overline{u} \end{array} \right) = \left( \begin{array}{ll} 1-\gamma &   \sum_{j=1}^N |f_j \rangle \langle \bar{f}_j|   \\ \sum_{j=1}^N |\bar{f}_j \rangle \langle f_j| & 1- \bar{\gamma} \end{array} \right) \end{equation}
is an implementable Bogoliubov transformation, with
\begin{equation}\label{eq:Gnu} \Gamma_\nu =  \left( \begin{array}{ll} \gamma & 0  \\ 0 & 1- \bar{\gamma} \end{array} \right). \end{equation}
Both the pure quasi-free state $R_\nu \Omega$ and the $N$-particle Slater determinant $\psi_{\text{slater}} (\bx) = (N!)^{-1/2}\det \left( f_j (x_i) \right)_{i,j \leq N}$ satisfy the Wick theorem and are therefore fully characterized by their generalized one-particle density. Since (\ref{eq:Gnu}) coincides with the generalized one-particle density of $\psi_{\text{slater}}$, it follows that $R_\nu \Omega = \{ 0, \dots , 0 , \psi_{\text{slater}}, 0 ,\dots \}$. Hence, Slater determinants are the only pure quasi-free states with vanishing pairing density.

\medskip

Although we will not make use of this fact, let us notice that unitary implementors of Bogoliubov transformations of the form (\ref{eq:Nu-gamma}), generating Slater determinants, can be conveniently constructed as particle-hole transformations. In fact, let $\{f_i\}_{i=1}^N$ be an orthonormal system on $L^2(\Rbb^3)$ and extend it to an orthonormal basis $\{f_i\}_{i=1}^\infty$. The Slater determinant $\psi_{\text{slater}} (\bx) = (N!)^{-1/2} \, \det (f_i (x_j))_{1\leq i,j \leq N}$ can be expressed as $R_\nu \Omega$, where the operator $R_\nu$ is defined by 
\[R_\nu\Omega := a^*(f_1)\cdots a^*(f_N) \Omega\]
and by the property 
\[ R_\nu a^*(f_i) R_\nu^* := \left\{\begin{array}{cc} a(f_i) & \text{ for } i \leq N \\ a^*(f_i) & \text{ for } i > N.\end{array}\right.\]
It is then simple to check that $R_\nu$ preserves the canonical anticommutation relations and that it is surjective. As a consequence, it is a unitary operator on $\cF$. 

\medskip

The next theorem is our main result. In it, we study the time evolution of initial data close to Slater determinants, and prove that their dynamics can be described in terms of the Hartree-Fock (or the Hartree) equation. Of course, we cannot start with an arbitrary Slater determinant. Instead, we need the initial state to have the semiclassical structure discussed in the introduction. We encode this requirement in the assumption (\ref{eq:sc}) below. We do not expect the result to be correct if the initial data is not semiclassical, i.e. if (\ref{eq:sc}) is not satisfied. Notice, however, that the semiclassical structure emerges naturally, when one tries to minimize the energy. Hence, the assumption (\ref{eq:sc}) is appropriate to study the dynamics of initially trapped fermionic systems close to the ground state of the trapped Hamiltonian (traps are then released, or changed, to observe the dynamics of the particles, which would otherwise be trivial). 
\begin{thm}
\label{thm:main}
Assume that, in the Hamiltonian (\ref{eq:cHN}), $V \in L^1 (\bR^3)$ is so that
\begin{equation}\label{eq:ass-V} \int \di p\, (1+|p|)^{2} |\widehat{V} (p)| < \infty \,. \end{equation}
Let $\omega_{N}$ be a sequence of orthogonal projections on $L^2 (\bR^3)$, with $\tr \omega_{N} = N$ and such that
\begin{equation}\label{eq:sc} 
\begin{split}
\tr \, |[ e^{i p \cdot x} , \omega_{N} ]| &\leq C N \hbar \, (1+|p|) \quad \text{and } \\
\tr \, |[ \hbar \nabla , \omega_{N} ]| &\leq C N \hbar 
\end{split} 
\end{equation}
for all $p \in \bR^3$ and for a constant $C >0$. Let $\nu_N$ denote the sequence of Bogoliubov transformations constructed in (\ref{eq:Nu-gamma}) such that $R_{\nu_N} \Omega$ has the generalized one-particle density 
\[ \Gamma_{\nu_N} = \left( \begin{array}{ll} \omega_N & 0 \\ 0 & 1- \overline{\omega}_N \end{array} \right). \]
Let $\xi_N \in \cF$ be a sequence with $\langle \xi_N , \cN \xi_N \rangle \leq C$ uniformly in $N$. Let $\gamma_{N,t}^{(1)}$ be the reduced one-particle density associated with the evolved state 
\begin{equation}\label{eq:psiNt} \psi_{N,t} = e^{-i \cH_N t/\hbar} R_{\nu_N} \xi_N \end{equation}
where the Hamiltonian $\cH_N$ has been defined in (\ref{eq:cHN}). On the other hand, denote by $\omega_{N,t}$ the solution of the Hartree-Fock equation 
\begin{equation}\label{eq:hf} i\hbar \partial_t \omega_{N,t} = \left[ -\hbar^2 \Delta + (V*\rho_t) - X_t , \omega_{N,t} \right] , \end{equation} with the initial data $\omega_{N,t=0} = \omega_N$. Here $\rho_t (x) = N^{-1} \omega_{N,t} (x;x)$ is the normalized density and $X_t$ is the exchange operator associated with $\omega_{N,t}$, having the kernel $X_t (x;y) = N^{-1} V(x-y) \omega_{N,t} (x;y)$. Then there exist constants $K, c_1, c_2 > 0$ such that  
\begin{equation}\label{eq:conv-HS} \left\| \gamma^{(1)}_{N,t} - \omega_{N,t} \right\|_{\text{HS}} \leq K \exp (c_2 \exp (c_1 |t|)) \end{equation}
and
\begin{equation}\label{eq:conv-tr0}
\tr \; \left|  \gamma^{(1)}_{N,t} - \omega_{N,t} \right| \leq K N^{1/2} \exp (c_2 \exp (c_1 |t|)) 
\end{equation}
for all $t \in \bR$. 

\medskip

Assume additionally that $d\Gamma (\omega_N) \xi_N = 0$ and $\langle \xi_N, \cN^2 \xi_N \rangle \leq C$ for all $N \in \bN$. Then 
there exist constants $K,c_1,c_2 > 0$ such that 
\begin{equation}\label{eq:conv-tr}
\tr \, \left| \gamma^{(1)}_{N,t} - \omega_{N,t} \right| \leq K N^{1/6} \exp (c_2 \exp (c_1 |t|)) 
\end{equation}
for all $t \in \bR$. Moreover, under this additional assumption, we obtain that 
\begin{equation}\label{eq:conv-sc} 
\left| \tr\, e^{i x \cdot q + \hbar p \cdot \nabla} \, \left( \gamma^{(1)}_{N,t} - \omega_{N,t} \right) \right| \leq K (1+|q|+|p|)^{1/2} \exp (c_2 \exp (c_1 |t|)) \end{equation}
for every $q,p \in \bR^3$, $t \in \bR$.  
\end{thm}

{\it Remarks.}
\begin{itemize}
\item Using (\ref{eq:Bog}), it is simple to check that $R_{\nu_N}^* \cN R_{\nu_N} = \cN - 2 d\Gamma (\omega_N) + N$. The assumption $\langle \xi_N, \cN \xi_N \rangle \leq C$ implies therefore that 
\[ \left| \tr\, \gamma_{N,0}^{(1)}  - \tr\, \omega_N \right| = \left| \langle \xi_N, R_{\nu_N}^* \cN R_{\nu_N} \xi_N \rangle - N \right| \leq C \]
uniformly in $N$ (this bound is of course preserved by the time-evolution). Following the arguments of Section \ref{sec:proof} it is also easy to check that \[ \| \gamma_{N,0}^{(1)}- \omega_N \|_{\text{HS}} \leq C, \quad \text{and } \tr \, | \gamma_{N,0}^{(1)} - \omega_N | \leq C N^{1/2}, \] if $\langle \xi_N, \cN \xi_N \rangle \leq C$. Under the additional assumption $d\Gamma (\omega_N) \xi_N = 0$, one can even show that \[ \tr \, \left| \gamma_{N,0}^{(1)} - \omega_N \right| \leq C, \]
uniformly in $N$ (applying the arguments at the beginning of Step 3 in Section \ref{sec:proof}). This proves that, at time $t=0$, the bulk of the particles is in the quasi-free state generated by $R_{\nu_N}$. The small fluctuations around the quasi-free state are described by $\xi_N$. In particular, it follows that the bounds (\ref{eq:conv-HS}), (\ref{eq:conv-tr0}), (\ref{eq:conv-tr}) and (\ref{eq:conv-sc}) hold at time $t=0$. Results similar to (\ref{eq:conv-HS}), (\ref{eq:conv-tr0}), (\ref{eq:conv-tr}), (\ref{eq:conv-sc}) also hold if $\langle \xi_N, \cN \xi_N \rangle \simeq N^\alpha$ and $\langle \xi_N, \cN^2 \xi_N \rangle \simeq N^\beta$, for some $\alpha,\beta > 0$, but then, of course, the errors become larger. 
\item Suppose that the initial data is $\omega_N = \sum_{j=1}^N |f_j \rangle \langle f_j|$ for a family $\{ f_j \}_{j=1}^N$ of orthonormal functions in $L^2 (\bR^3)$. 
Then the condition $d\Gamma (\omega_N) \xi_N = 0$, required for (\ref{eq:conv-tr}) and (\ref{eq:conv-sc}), is satisfied if $a (f_i) \xi_N = 0$ for all $i = 1, \dots , N$, meaning that particles in $\xi_N$ are orthogonal to all orbitals $f_{j}$ building the quasi-free part of the state. 
\item All our results and our analysis remain valid if we included an external potential in the Hamiltonian (\ref{eq:cHN}) generating the time-evolution (in this case, the external potential would, of course, also appear in the Hartree-Fock equation (\ref{eq:hf})).  
\item Eq. (\ref{eq:conv-HS}) is optimal in its $N$ dependence (it is easy to find a sequence $\xi_N \in \cF$ with $\langle \xi_N, \cN \xi_N \rangle < \infty$ such that, already at time $t=0$, the difference between $\gamma_{N,0}^{(1)}$ and $\omega_{N,0}$ is of order one).  On the other hand, we do not expect \eqref{eq:conv-tr0} and \eqref{eq:conv-tr} to be optimal (the optimal bound for the trace norm of the difference should be, like (\ref{eq:conv-HS}), of order one in $N$). Since the Hilbert-Schmidt norm of $\gamma_{N,t}^{(1)}$ and of $\omega_{N,t}$ is of the order $N^{1/2}$ (while their trace-norm is of order $N$), it is not surprising that in (\ref{eq:conv-HS}) we get a better rate than in \eqref{eq:conv-tr0} and in (\ref{eq:conv-tr}).
We point out, however, that we can improve (\ref{eq:conv-tr}) and get optimal estimates, if we test the
difference $\gamma^{(1)}_{N,t} - \omega_{N,t}$ against observables having the correct semiclassical structure, even if these observables are not Hilbert-Schmidt; see (\ref{eq:conv-sc}).
\item The bounds (\ref{eq:conv-HS}), (\ref{eq:conv-tr0}), (\ref{eq:conv-tr}), (\ref{eq:conv-sc}) deteriorate quite fast in time. The emergence of a double exponential is a consequence of the fact that when we propagate (\ref{eq:sc}) along the solution $\omega_{N,t}$ of the Hartree-Fock equation (\ref{eq:hf}) we get an additional factor which is growing exponentially in time. It is reasonable to expect that in many situations, the exponential growth for the commutators $[e^{ip \cdot x}, \omega_{N,t}]$ and $[\hbar \nabla , \omega_{N,t}]$ is too pessimistic. In these situation, it would be possible to get better time-dependences on the r.h.s. of (\ref{eq:conv-HS}),  
(\ref{eq:conv-tr0}), (\ref{eq:conv-tr}) and (\ref{eq:conv-sc}).
\item Let $\wt{\omega}_{N,t}$ denote the solution of the Hartree equation 
\begin{equation}\label{eq:h} i \hbar \partial_t \wt{\omega}_{N,t} = \left[ -\hbar^2 \Delta + (V * \wt{\rho}_t) , \wt{\omega}_{N,t} \right] \end{equation}
with the initial data $\omega_N$. Under the assumptions of Theorem \ref{thm:main} on the initial density $\omega_N$ and on the interaction potential $V$, we show in Appendix \ref{app:hartree} that the contribution of the exchange term $[X_t, \omega_{N,t}]$ in the Hartree-Fock equation (\ref{eq:hf}) is of smaller order, and that
\[ \tr |\omega_{N,t} - \wt{\omega}_{N,t} | \leq C \exp (c_1 \exp (c_2 |t|)). \]
It follows from this remark that the bounds (\ref{eq:conv-HS}), (\ref{eq:conv-tr0}), (\ref{eq:conv-tr}) and (\ref{eq:conv-sc}) remain true if we replace the solution $\omega_{N,t}$ of the Hartree-Fock equation with the solution $\wt{\omega}_{N,t}$ of the Hartree equation (with the same initial data).
\end{itemize}

\medskip

We can also control the convergence of higher order reduced densities. The $k$-particle reduced density associated with the evolved Fock state $\psi_{N,t}$ defined in (\ref{eq:psiNt}) is defined as the non-negative trace class operator $\gamma^{(k)}_{N,t}$ on $L^2 (\bR^{3k})$ with integral kernel given by
\[ \gamma^{(k)}_{N,t} (x_1,\ldots x_k ; x'_1,\ldots x'_k) = \big\langle \psi_{N,t} , a_{x'_1}^* \dots a_{x'_k}^* a_{x_k} \dots a_{x_1} \psi_{N,t} \big\rangle. \]
The $k$-particle reduced density associated with the evolved quasi-free state with one-particle density $\omega_{N,t}$ (obtained through the solution of the Hartree-Fock equation (\ref{eq:hf})) is given, according to Wick's theorem, by 
\begin{equation}\label{eq:omk} \omega^{(k)}_{N,t} (x_1,\ldots x_k ; x'_1,\ldots x'_k) = \sum_{\pi \in S_k} \sgn(\pi) \prod_{j=1}^k \omega_t (x_{j}; x'_{\pi(j)}). \end{equation} 
Notice that these reduced densities are normalized such that $\tr \, \omega_{N,t}^{(k)} = N!/ (N-k)!$. 
\begin{thm}\label{thm:k}
We use the same notations and assume the same conditions as in Theorem~\ref{thm:main} (the condition $d\Gamma (\omega_N) \xi_N = 0$ is not required here). Let $k \in \bN$ and assume, additionally, that the sequence $\xi_N$ is such that $\langle \xi_N , (\cN+1)^{k} \xi_N \rangle \leq C$. Then there exists constants $D, c_1, c_2 > 0$ (with $c_1$ depending only on $V$ and on the constant on the r.h.s. of (\ref{eq:sc}) and $D,c_2$ depending on $V$, on the constants on the r.h.s. of (\ref{eq:sc}) and on $k$) such that
\begin{equation}\label{eq:k-claim} \left\| \gamma^{(k)}_{N,t} - \omega^{(k)}_{N,t} \right\|_{\text{HS}} \leq D N^{(k-1)/2} \exp (c_2 \exp (c_1 |t|)) \end{equation}
and 
\begin{equation}\label{eq:k-claim-tr} \tr \left| \gamma^{(k)}_{N,t} - \omega^{(k)}_{N,t} \right| \leq D N^{k-\frac{1}{2}}  \, \exp (c_2 \exp (c_1 |t|)). \end{equation}
\end{thm}

{\it Remark.} The $N$-dependence of the bound (\ref{eq:k-claim}) is optimal. On the other hand, the $N$-dependence of (\ref{eq:k-claim-tr}) is not expected to be optimal, the optimal bound for the trace norm of the difference $\gamma^{(k)}_{N,t} - \omega_{N,t}^{(k)}$ should be of the order $N^{k-1}$. 


\medskip

In order to show Theorem \ref{thm:main} and Theorem \ref{thm:k} we are going to compare the fully evolved Fock state $\psi_{N,t} = e^{-i\cH_N t/\hbar} R_{\nu_N} \xi_N$ with the quasi-free state on $\cF$ with reduced one-particle density given by the solution $\omega_{N,t}$ of the Hartree-Fock equation (\ref{eq:hf}). To this end, we write $\omega_{N,t} = \sum_{j=1}^N |f_{j,t} \rangle \langle f_{j,t}|$ for an orthonormal family $\{ f_{j,t} \}_{j=1}^N$ in $L^2 (\bR^3)$. Notice that the functions $f_{j,t}$ can be  determined by solving the system of $N$ coupled non-linear equations
\[ \begin{split} i \hbar \partial_t f_{j,t} (x) = \; &- \hbar^2 \Delta f_{j,t} (x) + \frac{1}{N} \sum_{i=1}^N \int dy  V (x-y) |f_{i,t} (y)|^2 f_{j,t} (x) \\ &- \frac{1}{N} \sum_{i=1}^N \int dy V(x-y) f_{j,t} (y) \overline{f}_{i,t} (y) f_{i,t} (x) \end{split} \]
with the initial data $f_{j,t=0} = f_j$ appearing in (\ref{eq:Nu-gamma}). This system of equations is equivalent to the Hartree-Fock equation (\ref{eq:hf}). We define then $u_{N,t} = 1- \omega_{N,t}$ and $v_{N,t} = \sum_{j=1}^N |\cc f_{j,t} \rangle \langle f_{j,t}|$. Similarly to (\ref{eq:Nu-gamma}), we define the Bogoliubov transformation 
\begin{equation}\label{eq:nuNt} \nu_{N,t} = \left( \begin{array}{ll} u_{N,t} & \cc v_{N,t} \\ v_{N,t} & \cc u_{N,t} \end{array} \right) =  \left(\begin{array}{ll} 1-\omega_{N,t} & \sum_{j=1}^N |f_{j,t} \rangle \langle \cc f_{j,t}| \\ \sum_{j=1}^N |\cc f_{j,t} \rangle \langle f_{j,t}| & 1- \cc \omega_{N,t} \end{array} \right). \end{equation}
The generalized reduced density matrix associated with the quasi-free state $R_{\nu_{N,t}} \Omega$ is given by 
\[ \Gamma_{\nu_{N,t}} = \left( \begin{array}{ll} \omega_{N,t} & 0 \\ 0 & 1- \cc \omega_{N,t} \end{array} \right) \, . \]

\medskip

We expect $\psi_{N,t}$ to be close to the quasi-free state $R_{\nu_{N,t}} \Omega$. To prove that this is indeed the case, we define $\xi_{N,t} \in \cF$ so that
\[ \psi_{N,t} = e^{-i\cH_N t/\hbar} R_{\nu_N} \xi_N = R_{\nu_{N,t}} \xi_{N,t} \]
for every $t \in \bR$. Equivalently, $\xi_{N,t} = \cU_N (t;0) \xi_N$, where we defined the two-parameter group of unitary transformations 
\begin{equation}\label{eq:cU} \cU_N (t;s) = R_{\nu_{N,t}}^* e^{-i\cH_N (t-s)/\hbar} R_{\nu_{N,s}} \end{equation}
for any $t,s \in \bR$. We refer to $\cU_N$ as the fluctuation dynamics; it describes the evolution of particles which are outside the quasi-free state. 

\medskip

As we will show in detail in Section \ref{sec:proof}, the problem of proving the convergence of $\gamma_{N,t}^{(1)}$ towards the solution of the Hartree-Fock equation $\omega_t$ can be reduced to the problem of controlling the expectation of the number of particles operator (and of its powers) in the state $\xi_{N,t}$, or, equivalently, of controlling the growth of the number of particles operator with respect to the fluctuation dynamics $\cU_N$. 

\medskip

In spirit, this approach is similar to the coherent state method developed in \cite{RS} for bosonic systems. In the bosonic case, however, one considers the evolution of approximately coherent initial data. In this case, the fluctuation dynamics is obtained by conjugating the full evolution $\exp (-i\cH_N (t-s))$ with evolved Weyl operators, in contrast to the Bogoliubov transformation appearing in (\ref{eq:cU}). Notice that also in the bosonic case, Bogoliubov transformations can be applied, in addition to Weyl operators, to obtain a more precise description of the evolution. In particular, bosonic Bogoliubov transformations have been used in \cite{BKS}, to describe fluctuations around the condensate, and in \cite{BDS}, to implement the short scale correlation structure produced by the singular potential. In both cases, bosonic Bogoliubov transformations describe corrections to the evolving condensate created by the Weyl operator. In the fermionic case, on the other hand, the Pauli principle excludes the presence of a condensate and the Bogoliubov transformations produce the main term in the approximation of the many body dynamics. 

\section{Bounds on growth of fluctuations}\label{sec:gro}
\setcounter{equation}{0}

In this section we prove bounds for the growth of the expectation of the number of particles operator and of its powers with respect to the fluctuation dynamics $\cU_N (t;s)$. To obtain such estimates, we will make use of the following lemma, where we collect a series of important bounds for operators on the fermionic Fock space.
\begin{lem}\label{lm:bds-2}
For every bounded operator $O$ on $L^2 (\bR^3)$, we have
\[ \| d\Gamma (O) \psi \| \leq \| O \| \, \| \cN \psi \| \]
for every $\psi \in \cF$. If $O$ is a Hilbert-Schmidt operator, we also have the bounds
\begin{equation} \label{eq:HS-bds} \begin{split} \| d\Gamma (O) \psi \| &\leq \| O \|_{\text{HS}} \, \| \cN^{1/2} \psi \|, \\ 
\left\| \int dx dx' \, O(x;x') a_x a_{x'} \psi \right\| &\leq \| O \|_{\text{HS}} \, \| \cN^{1/2} \psi \|, \\
\left\| \int dx dx' \, O(x;x') a^*_x a^*_{x'} \psi \right\| & \leq 2 \| O \|_{\text{HS}} \, \| (\cN+1)^{1/2} \psi \| .
\end{split} \end{equation}
for every $\psi \in \cF$. Finally, if $O$ is a trace class operator, we obtain
\begin{equation}\label{eq:tr-bds} 
\begin{split} \| d\Gamma (O) \| &\leq 2 \| O\|_{\tr}\,, \\ 
\left\| \int dx dx' \, O(x;x') a_x a_{x'} \right\| &\leq 2 \| O \|_{\tr}\,,  \\
\left\| \int dx dx' \,  O(x;x') a_x^* a_{x'}^* \right\| &\leq 2 \| O \|_{\tr}  \, . 
\end{split} 
\end{equation}
Here $\| O \|_{\tr} = \tr \, |O| = \tr \sqrt{O^* O}$ indicates the trace norm of $O$.
\end{lem}

\begin{proof}
For any bounded operator $O$ on $L^2 (\bR^3)$ we have
\[ \| d\Gamma (O) \psi \|^2 = \sum_{n=1}^\infty \sum_{i,j=1}^n \langle \psi^{(n)}, O^{(i)} O^{(j)} \psi^{(n)} \rangle \leq \| O \|^2 \sum_{n=1}^\infty n^2 \| \psi^{(n)} \|^2 = \| O \|^2 \| \cN \psi \|^2 \,. \]
For a Hilbert-Schmidt operator $O$ on $L^2 (\bR^3)$, we have, using (\ref{eq:bdaa*}), 
\begin{equation}\label{eq:hs-bd} \begin{split} \left\| \int dx dx' \, O(x';x) \, a_{x'}^\# a_{x} \psi \right\| &\leq \int dx \, \| a^\# (O (.;x)) a_{x} \psi \| \\ &\leq \int dx\, \| O(.;x) \|_2 \, \| a_{x} \psi \|  \\ &\leq \| O \|_{\text{HS}} \left( \int dx \, \| a_x \psi \|^2 \right)^{1/2} \leq \| O \|_{\text{HS}} \| \cN^{1/2} \psi \|\end{split} \end{equation}
where $a^\#$ is either an annihilation operator $a$ or a creation operator $a^*$. This proves the first two bounds in (\ref{eq:HS-bds}).
The third bound in \eqref{eq:HS-bds} can be reduced to the previous bound as follows:
\begin{align*}
\left\| \int \di x\di y\, O(x;y) a^*_x a^*_y \psi \right\| & = \sup_{\varphi \in \fock,\ \norm{\varphi}=1} \left|\left\langle \varphi, \int \di x\di y\, O(x;y) a^*_x a^*_y \psi \right\rangle\right| \\
& = \sup_{\varphi \in \fock,\ \norm{\varphi}=1 } \left|\left\langle \int \di x\di y\, \overline{O(x;y)} a_x a_y (\Ncal+1)^{-1/2} \varphi,(\Ncal+3)^{1/2} \psi \right\rangle\right|\\
& \leq \sup_{\varphi \in \fock,\ \norm{\varphi}=1} \norm{O}\HS \norm{\Ncal^{1/2}(\Ncal+1)^{-1/2} \varphi} \norm{(\Ncal+3)^{1/2}\psi}\\& \leq \norm{O}\HS \norm{(\Ncal+3)^{1/2}\psi}.
\end{align*}

 Finally, we prove (\ref{eq:tr-bds}). Assume first that $O$ is a self-adjoint trace-class operator. Then we have the spectral decomposition
\[ O = \sum_j \lambda_j |f_j \rangle \langle f_j| \]
for a real sequence $\{ \lambda_j \}$ of eigenvalues with $\sum_j |\lambda_j| = \tr \, |O|$ and an orthonormal family of eigenvectors $f_j \in L^2 (\bR^3)$. We find
\[ \left\| \int dxdx' \, O (x;x') \, a_x^\# a_{x'}^\# \right\| \leq \sum_j |\lambda_j| \left\| \int dx dx' \,  f_j (x') \overline{f}_j (x) \, a_x^\# a_{x'}^\# \right\| = \sum_j |\lambda_j| \left\| a^\# (\wt{f}_j) a^\# (\wt{f}_j) \right\| \]
where $\wt{f}_j$ is either $f_j$ or its complex conjugate $\overline{f}_j$. We conclude from (\ref{eq:bdaa*})
that
\begin{equation}\label{eq:tr-sa} \left\|  \int dxdx' \, O (x;x') a_x^\# a_{x'}^\# \right\| \leq \sum_j |\lambda_j| \| f_j \|^2 = \| O \|_{\text{tr}} \, . \end{equation}
Now, for an arbitrary, not necessarily self-adjoint, trace-class operator $O$, we write
\[ O = \frac{O+O^*}{2} + i \frac{O-O^*}{2i} \, . \]
Therefore, applying (\ref{eq:tr-sa}), we find 
\[ \begin{split} \Big\| \int dx dx' \, & O (x;x') \, a_x^\# a_{x'}^\# \Big\| \\ \leq \; & \left\| \int dx dx'  \, \left(\frac{O+O^*}{2}\right) (x;x') \, a_x^\# a_{x'}^\# \right\| + \left\| \int dx dx' \, \left(\frac{O-O^*}{2i}\right) (x;x') \, a_x^\# a_{x'}^\# \right\| \\ \leq \; & \left\| \frac{O+O^*}{2} \right\|_{\text{tr}} + \left\| \frac{O-O^*}{2i} \right\|_{\text{tr}} \leq 2 \| O \|_{\text{tr}} \, .\qedhere \end{split} \]
\end{proof}

We are now ready to state the main result of this section, which is a bound for the growth of the expectation of $(\cN+1)^k$ with respect to the fluctuation dynamics. 
\begin{thm}\label{thm:grow} 
Assume (\ref{eq:ass-V}) and (\ref{eq:sc}). Let $\cU_N (t;s)$ be the fluctuation dynamics defined in (\ref{eq:cU}) and $k \in \bN$. Then there exist a constant $c_1 > 0$, depending only on $V$, and a constant $c_2 > 0$ depending on $V$ and on $k$ such that 
\begin{equation}\label{eq:grow} \left\langle \xi, \cU_N (t;0)^* (\cN + 1)^k \, \cU_N (t;0) \xi \right\rangle \leq  \exp (c_2 \exp (c_1 |t|)) \, \langle \xi, (\cN + 1)^k \xi \rangle. \end{equation}
\end{thm}

\medskip

The first step in the proof of Theorem \ref{thm:grow} is an explicit computation of the time derivative of the expectation of the evolved moments of the number of particles operator appearing on the l.h.s. of (\ref{eq:grow}). Recall from (\ref{eq:cU}) that $\cU_N (t;0) = R^*_{\nu_{N,t}} e^{-i\cH_N t/\hbar} R_{\nu_N}$, where  
\[ \nu_{N,t} = \left( \begin{array}{ll} u_{N,t} & \cc v_{N,t} \\ v_{N,t} & \cc u_{N,t} \end{array} \right) \]
is the Bogoliubov transform defined in (\ref{eq:nuNt}), with $v_{N,t}^* v_{N,t} = \omega_{N,t}$ and $u_{N,t} = 1-\omega_{N,t}$.

In the rest of this section, we will use the shorthand notation $R_t \equiv R_{\nu_{N,t}}$, $u_t \equiv u_{N,t}$, $v_t \equiv v_{N,t}$ and $\cc v_t \equiv \cc v_{N,t}$. Moreover, we define the functions $u_{t,x}, v_{t,x}, \cc v_{t,x}$ by $u_{t,x} (y)= u_{N,t} (y;x)$, $v_{t,x} (y)= v_{N,t}  (y;x)$ and $\cc v_{t,x} (y) = \cc v_{N,t} (y;x)$, where $u_{N,t} (y;x)$, $v_{N,t} (y;x)$ and $\cc v_{N,t} (y;x)$ denote the integral kernels of the operators $u_{N,t}$, $v_{N,t}$ and $\cc v_{N,t}$. Note that, from \eqref{eq:Bog}, the action of the Bogoliubov transformation $R_t$ on the operator valued distributions $a_x, a_x^*$ is given by
\[R^*_t a_x R_t = a(u_{t,x})  + a^* (\cc v_{t,x})\quad \text{and } \quad R^*_t a^*_x R_t = a^* (u_{t,x}) +a (\cc v_{t,x})  \, .  \]
\begin{prp}\label{prop:dt}
Let $\cU_N (t;s)$ be the fluctuation dynamics defined in (\ref{eq:cU}), $\xi \in \cF$, and $k \in \bN$. Then 
\begin{equation}\label{eq:prop-dt} \begin{split} 
i \hbar \frac{d}{dt} \, \Big\langle &\cU_N (t;0) \xi , (\Ncal+1)^k \cU_N (t;0) \xi\Big\rangle \\ = \; -&\frac{4 i}{N} \Im \sum_{j=1}^k \int  dx dy\, V(x-y) \\ &\times \left\{ \left\langle \cU_N (t;0) \xi, (\Ncal+1)^{j-1} a^*(u_{t,x}) a(\cc v_{t,y}) a(u_{t,y}) a(u_{t,x}) (\Ncal+1)^{k-j} \cU_N(t;0)\xi\right\rangle \right. \\ & \quad +\left\langle \cU_N (t;0) \xi, (\Ncal+1)^{j-1}  a(\cc v_{t,x}) a(\cc v_{t,y}) a(u_{t,y}) a(u_{t,x}) (\Ncal+1)^{k-j} \cU_N (t;0)\xi\right\rangle \\ &\quad \left. + \left\langle \cU_N (t;0) \xi, (\Ncal+1)^{j-1} a^*(u_{t,y}) a^*(\cc v_{t,y}) a^*(\cc v_{t,x}) a(\cc v_{t,x}) (\Ncal+1)^{k-j} \cU_N (t;0)\xi\right\rangle \right\}. \end{split} \end{equation}
\end{prp}
%

\begin{proof}
A simple computation using (\ref{eq:Bog}) shows that
\[ R_t \, \cN R^*_{t} = \cN - 2 d\Gamma (\omega_{N,t}) + N \] 
and therefore that 
\[ \begin{split} \cU_N (t;0)^* \cN \cU_N (t;0)  & = R_0^* \, e^{i\cH_N t/\hbar} (\cN - 2 d\Gamma (\omega_{N,t}) + N) e^{-i \cH_N t/\hbar} R_0 \\ &
= R_0^* \, \cN R_0 -2 R_0^* e^{i\cH_N t/\hbar} d\Gamma (\omega_{N,t}) e^{-i\cH_N t/\hbar} R_0 + N. \end{split}\]
Hence
\[ \begin{split} 
i\hbar \frac{d}{dt} \, \cU_N^* (t;0) \cN \cU_N (t;0) &= -2 \, R_0^* \, e^{i\cH_N t/\hbar} \left\{ d\Gamma (i\hbar \partial_t \omega_{N,t}) - [\cH_N , d\Gamma (\omega_{N,t}) ] \right\} e^{-i \cH_N t/\hbar} R_0  \\ &= -2 \, \cU_N^* (t;0) R_t^* \left\{  d\Gamma (i\hbar \partial_t \omega_{N,t}) - [\cH_N , d\Gamma (\omega_{N,t}) ] \right\} R_t \,\cU_N (t;0). \end{split} \]
On the one hand, from the Hartree-Fock equation (\ref{eq:hf}) for $\omega_{N,t}$ we find
\[ d\Gamma (i\hbar \partial_t \omega_{N,t}) =  d\Gamma \left( \left[ -\hbar^2 \Delta , \omega_{N,t} \right] \right) + d\Gamma \left( \left[ V * \rho_t -X_t ,\omega_{N,t} \right] \right) \]
where we recall the definitions of the normalized density $\rho_t (x) = (1/N) \omega_{N,t} (x;x)$ and of the exchange operator $X_t (x;y) = (1/N) V(x-y) \omega_{N,t} (x;y)$. On the other hand
\[ [\cH_N , d\Gamma (\omega_{N,t})] = [ d\Gamma (-\hbar^2 \Delta) , d\Gamma (\omega_{N,t}) ] + [ \Vcal_N , d\Gamma (\omega_{N,t}) ] \]
with the interaction 
\[ \Vcal_N = \frac{1}{2N} \int dx dy V(x-y) a_x^* a_y^* a_y a_x. \]
We conclude that
\begin{equation}\label{eq:dtUNU}
\begin{split}  i\hbar \frac{d}{dt} \, &\cU_N^* (t;0) \cN \cU_N (t;0) \\ & = -2 \cU_N^* (t;0) R_t^* \left\{ d\Gamma \left( [ V* \rho_t - X_t , \omega_{N,t}] \right) - [ \Vcal_N , d\Gamma (\omega_{N,t})] \right\} R_t \cU_N (t;0). \end{split} \end{equation}
Next, we compute the two terms in the brackets. The first term is given by
\begin{equation} \label{canc4}
\begin{split}
d\Gamma(&[V\ast \rho_{t} - X_{t},\omega_{N,t}]) \\ = &\; \frac{1}{N}\int \di z_1 \di z_2\,a^*_{z_1} a_{z_2}  \int \di x\, V(z_1-x)\big[ \omega_{N,t}(z_1;z_2)\omega_{N,t}(x;x) - \omega_{N,t}(z_1;x)\omega_{N,t}(x;z_2) \big] \\ &- \hc 
\end{split} \end{equation}
Using \eqref{eq:Bog}, we find
\begin{equation}
\begin{split}
R^*_t  & d\Gamma([V\ast \r_{t} - X_{t},\omega_{N,t}]) R_t \\ & = \frac{1}{N} \int \di z_1 \di z_2\, \left(a^*(u_{t,z_1}) + a(\cc v_{t,z_1})\right)\left(a(u_{t,z_2}) + a^*(\cc v_{t,z_2})\right) \\ &\quad\quad \times \int  dx \, V(z_1-x) \big[ \omega_{N,t}(z_1;z_2)\omega_{N,t}(x;x) - \omega_{N,t}(z_1;x)\omega_{N,t}(x;z_2) \big] - \hc
\end{split}
\end{equation}
The integration over $z_2$ can be done explicitly using the property $\int  dz_2 \,  \cc u_t (y_1;z_2) \omega_{N,t}(y_2;z_2) = (\omega_{N,t} u_t)(y_2;y_1) = 0$ and the fact that $\omega_{N,t} \cc v_t  = \cc v_t$. We get
\begin{equation}\begin{split}
\label{eq:hf-term}
R^*_t  d\Gamma &([V\ast \r_{t} - X_{t},\omega_{N,t}]) R_t \\
& = \frac{1}{N} \int \di z_1\, \left(a^*(u_{t,z_1}) + a(\cc v_{t,z_1})\right) \\&\qquad \times\int \di x\, V(z_1-x)\big[ a^*(\cc v_{t,z_1})\omega_{N,t}(x;x) - a^*(\cc v_{t,x})\omega_{N,t}(z_1;x)\big] - \hc\\
& = \frac{1}{N} \int \di x \di y\,V(x-y)\bigg[ \omega_{N,t}(x;x) a^*(u_{t,y}) a^*(\cc v_{t,y}) - \omega_{N,t}(y;x) a^*(u_{t,y}) a^*(\cc v_{t,x}) \bigg] - \hc
\end{split} \end{equation}
where in the last step the contributions containing $a(\cc v_{t,z_1})$ are cancelled by their hermitian conjugates. 

\medskip

We now consider the second contribution in the brackets on the r.h.s. of (\ref{eq:dtUNU}). Using the canonical anticommutation relations,  we obtain
\[[\Vcal_N,\di\Gamma(\omega_{N,t})] = \frac{1}{N} \int \di x\di y\di z\, V(x-y) \omega_{N,t}(z,y) a^*_z a^*_x a_y a_x - \hc\]
Conjugating with the Bogoliubov transformation $R_t$, we find
\begin{align*}
 R^*_t &[\Vcal_N,\di\Gamma(\omega_{N,t})] R_t\\
 & = \frac{1}{N}\int \di x \di y \di z\,V(x-y) \omega_{N,t}(z;y) \\ & \qquad \times \left(a^*(u_{t,z})+a(\cc v_{t,z})\right)\left(a^*(u_{t,x})+a(\cc v_{t,x})\right)
 \left(a(u_{t,y})+a^*(\cc v_{t,y})\right)\left(a(u_{t,x})+a^*(\cc v_{t,x})\right) \\ &\quad -\hc
\end{align*}
Integrating over $z$, using again $\omega_{N,t} u_t = 0$ and $\omega_{N,t} \cc v_t = \cc v_t$, we find
\[ \begin{split}
R^*_t &[\Vcal_N,\di\Gamma(\omega_{N,t})] R_t\\
 & = \frac{1}{N}\int \di x \di y \,V(x-y)a(\cc v_{t,y})\left(a^*(u_{t,x})+a(\cc v_{t,x})\right)\left(a(u_{t,y})+a^*(\cc v_{t,y})\right)\left(a(u_{t,x})+a^*(\cc v_{t,x})\right) \\ &\qquad -\hc
 \end{split} \]
 Since $\scal{\cc v_{t,y}}{u_{t,x}}=0$ the operators $a(\cc v_{t,y})$ and $a^*(u_{t,x})$ anticommute. Taking into account the fact that many contributions cancel after subtracting the hermitian conjugate, we find
\[ \begin{split} 
R^*_t [\Vcal_N, &\di\Gamma(\omega_{N,t})] R_t \\  
= &\; -\frac{1}{N}\int \di x \di y \,V(x-y) \,  \bigg[a^*(u_{t,x})a(\cc v_{t,y}) a(u_{t,y}) a(u_{t,x}) \\ 
& \hspace{1cm} +a(\cc v_{t,x}) a(\cc v_{t,y}) a(u_{t,y}) a(u_{t,x})
 -a(\cc v_{t,x}) a(\cc v_{t,y}) a^*(\cc v_{t,x}) a(u_{t,y}) \bigg] -\hc
\end{split} \]
Normal ordering the last term in the brackets using $\scal{\cc v_{t,y}}{\cc v_{t,x}} = \omega_{N,t}(x;y)$, we conclude that 
\begin{equation}\label{quantumterm} \begin{split}
R^*_t &[\Vcal_N,\di\Gamma(\omega_{N,t})] R_t\\
   = \; & -\frac{1}{N}\int \di x \di y \,V(x-y) \, \bigg[a^*(u_{t,x})a(\cc v_{t,y}) a(u_{t,y}) a(u_{t,x})  \\& \hspace{2cm} +a(\cc v_{t,x}) a(\cc v_{t,y}) a(u_{t,y}) a(u_{t,x})
 +a^*(u_{t,y}) a^*(\cc v_{t,y})a^*(\cc v_{t,x}) a(\cc v_{t,x})  \bigg]  -\hc\\
&+\frac{1}{N}\int \di x \di y \,V(x-y)\bigg[\omega_{N,t}(x;x) a^*(u_{t,y})a^*(\cc v_{t,y}) -\omega_{N,t}(y;x)a^*(u_{t,y})a^*(\cc v_{t,x})  \bigg] -\hc
\end{split}\end{equation}
Combining (\ref{eq:hf-term}) with (\ref{quantumterm}), we find 
\[ \begin{split} R_t^* &\left\{ d\Gamma \left( [V* \rho_t - X_t , \omega_{N,t}] \right) - [ \Vcal_N, d\Gamma (\omega_{N,t})] \right\} R_t \\ = \; &- \frac{1}{N}\int \di x \di y \,V(x-y) \, \bigg[ a^*(u_{t,x})a(\cc v_{t,y}) a(u_{t,y}) a(u_{t,x})
\\& \hspace{1cm}
+a(\cc v_{t,x}) a(\cc v_{t,y}) a(u_{t,y}) a(u_{t,x})
+a^*(u_{t,y}) a^*(\cc v_{t,y})a^*(\cc v_{t,x}) a(\cc v_{t,x})  \bigg] -\hc
\end{split} \]
{F}rom (\ref{eq:dtUNU}), we obtain
\[ \begin{split}
i\hbar \frac{d}{dt} \, &\cU_N^* (t;0) \cN \cU_N (t;0) \\ =\; & -\frac{4i}{N} \, \im \int dx dy \, V(x-y) \, \cU_N^* (t;0) \bigg[ a^*(u_{t,x})a(\cc v_{t,y}) a(u_{t,y}) a(u_{t,x})
\\ &\hspace{1cm}
+a(\cc v_{t,x}) a(\cc v_{t,y}) a(u_{t,y}) a(u_{t,x})+a^*(u_{t,y}) a^*(\cc v_{t,y})a^*(\cc v_{t,x}) a(\cc v_{t,x})  \bigg] \cU_N (t;0).
\end{split} \]
Eq. (\ref{eq:prop-dt}) now follows from the observation that 
\[ \begin{split} 
i\hbar & \frac{d}{dt}  \Big\langle \xi, \cU_N^* (t;0) (\cN+1)^k \cU_N (t;0) \xi \Big\rangle \\ & = \sum_{j=1}^k  \Big\langle \xi,  \cU_N^* (t;0) (\cN+1)^{j-1}  \\ &\hspace{1cm} \times \cU_N (t;0) \left[ i\hbar \frac{d}{dt} \cU_N^* (t;0) \cN \cU_N (t;0) \right] \cU_N^* (t;0) (\cN+1)^{k-j} \cU_N (t;0) \xi \Big\rangle. \qedhere
\end{split} \]
\end{proof}

\medskip

Next, we have to bound the three terms on the r.h.s. of (\ref{eq:prop-dt}) by the expectation of $(\cN+1)^k$ in the state $\cU_N (t;0) \xi$. A key ingredient to obtain such bounds is an estimate for the trace norm of the commutators $[e^{ip \cdot x}, \omega_{N,t}]$. For $t=0$ such an estimate was assumed in (\ref{eq:sc}). In the next proposition, whose proof is deferred to Section \ref{sec:sc}, we show that the bound can be propagated to all $t \in \bR$.
\begin{prp}
\label{lem:hbargain} 
Let $V \in L^1 (\bR^3)$ such that \[ \int dp \, (1+|p|^2) \, |\widehat{V} (p)| < \infty \, . \] Let $\omega_N$ be a non-negative trace class operator on $L^2 (\bR^3)$, with $\tr \omega_N = N$, $\|\o_{N}\|\leq 1$ and such that
\be
\begin{split}
\sup_{p \in \bR^3} \frac{1}{1+|p|} \tr |[\omega_N,e^{ip\cdot x}]| &\leq CN\hbar \\  \tr |[\omega_N,\hbar\nabla]| &\leq CN\hbar\;.\label{h1}
\end{split}
\ee
for all $p \in \bR^3$. Let $\omega_{N,t}$ be the solution of the Hartree-Fock equation (\ref{eq:hf}) with initial data $\omega_N$. Then, there exist constants $K,c>0$ only depending on the potential $V$ such that 
\be
\begin{split}
\sup_{p\in \bR^3} \frac{1}{1+|p|} \tr |[\omega_{N,t},e^{ip \cdot x}]| &\leq K N\hbar \, \exp (c|t|) \\ \tr|[\omega_{N,t},\hbar\nabla]| &\leq K N\hbar \, \exp (c|t|) 
\label{h1b}
\end{split}
\ee
for all $p \in \bR^3$ and $t \in \bR$. 
\end{prp}

\medskip

We are now ready to estimate the three terms appearing on the r.h.s. of (\ref{eq:prop-dt}).  
\begin{lem}
\label{lem:errorest} Under the assumptions (\ref{eq:ass-V}) und (\ref{eq:sc}) of Theorem \ref{thm:main}, there exists a constant $c_1 > 0$ depending on $V$ and a constant $C>0$ depending on $V$ and on $k \in \bN$, such that 
\begin{equation}\label{eq:erroest}
\begin{split}
\Big| &\frac{1}{N} \int dx dy\, V (x-y) \, \Big\langle \cU_N (t;0) \xi , (\cN+1)^{j-1} \, \bigg\{ a^* (u_{t,x}) a (\cc v_{t,y}) a (u_{t,y}) a(u_{t,x}) \\ &+ a (\cc v_{t,x}) a (\cc v_{t,y}) a (u_{t,y}) a(u_{t,x}) +  a^*(u_{t,y}) a^*(\cc v_{t,y}) a^*(\cc v_{t,x}) a(\cc v_{t,x})  \bigg\} 
(\cN+1)^{k-j}  \cU_N (t;0) \xi \Big\rangle \Big| \\ &\hspace{5cm} \leq C \hbar \exp (c_1 |t|)  \, \big\langle \cU_N (t;0) \xi,  (\cN+1)^{k} \cU_N (t;0) \xi \big\rangle
\end{split} 
\end{equation}
for all $j=1, \dots , k$ and $t \in \bR$.
\end{lem}
\begin{proof}
We estimate the contributions arising from the three terms in the parenthesis separately. Let us start with the first term,
\[ \begin{split} \text{I} := \; & \Big| \frac{1}{N} \int dx dy\, V (x-y)\, \Big\langle \cU_N (t;0) \xi , (\cN+1)^{j-1} \\ &\hspace{2cm} \times a^* (u_{t,x}) a (\cc v_{t,y}) a (u_{t,y}) a(u_{t,x}) (\cN+1)^{k-j}  \cU_N (t;0) \xi \Big\rangle \Big|\;. \end{split} \]
Inserting $1 = (\cN + 3)^{k/2 -j}(\cN + 3)^{-k/2 + j}$, pulling $(\cN + 3)^{-k/2 + j}$ through the fermonic operators to the right, and using the Cauchy-Schwarz  inequality, we get:
\begin{equation}\label{eq:II} \begin{split}
\text{I} \leq\; & \frac{1}{N} \int dp |\widehat{V} (p)| \left\| \int dx\, a^* (u_{t,x}) e^{ip \cdot x} a (u_{t,x}) (\cN+3)^{k/2-j} (\cN+1)^{j-1} \cU_N (t;0) \xi \right\| \\ &\quad \times  \left\| \int dy \, a (\cc v_{t,y}) e^{-ip \cdot y} a (u_{t,y}) (\cN+1)^{k/2} \cU_N (t;0) \xi \right\|\;.  \end{split} \end{equation}
The first norm can be bounded using that, for any $\phi \in \fock$:
\begin{equation}\label{eq:IIbd1}  \begin{split} \left\|\int dx\, a^* (u_{t,x}) e^{ip \cdot x} a(u_{t,x})\phi\right\| &= \left\|\int dx dr_1 dr_2\, u_t (r_1,x) e^{ip\cdot x} u_t (x,r_2) a^*_{r_1} a_{r_2}\phi\right\| \\ &= \left\|d\Gamma (u_t e^{ip \cdot x} u_t)\phi\right\| \\ &\leq \left\| \cN \phi \right\| \end{split} \end{equation}
where the last line follows from Lemma \ref{lm:bds-2} together with $\|u_{t}e^{ip\cdot x}u_{t}\| \leq 1$ (with a slight abuse of notation, $e^{ip\cdot x}$ denotes a multiplication operator). As for the second norm on the r.h.s. of (\ref{eq:II}), we use that:
\bea
\left\|\int dy\, a (\cc v_{t,y}) e^{-ip\cdot y} a (u_{t,y})\phi\right\| &=& \left\|\int dr_1 dr_2 \left(v_t e^{-ip \cdot x}u_{t}\right)(r_1; r_2) a_{r_1} a_{r_2}\phi\right\|\nn\\
&=& \left\|\int dr_1 dr_2 \left(v_t [ e^{-ip \cdot x}, \omega_{N,t}] \right)(r_1; r_2) a_{r_1} a_{r_2}\phi\right\|\nn\\
&\leq& 2\left\|v_t [ e^{-ip \cdot x}, \omega_{N,t}] \right\|_{\text{tr}}\|\phi\|\nn\\
&\leq& 2K \hbar (1+|p|)  N e^{c |t|}\|\phi\|\label{eq:IIbd2}
\eea
where the second line follows from $v_{t}u_{t}=0$ and $u_{t} = 1 - \o_{N,t}$, the third from from Lemma~\ref{lm:bds-2}, while the last from $\|v_{t}\|\leq 1$ and Proposition \ref{lem:hbargain}. Using the bounds (\ref{eq:IIbd1}), (\ref{eq:IIbd2}) in (\ref{eq:II}) we get:
\bea
\text{I} &\leq& 2K\hbar \left( \int dp |\widehat{V} (p)| (1+|p|) \right) e^{c|t|} \left\| \cN (\cN+3)^{k/2-j} (\cN+1)^{j-1} \cU_N (t;0) \xi \right\|\nn\\&&\quad\times \left\| (\cN +1)^{k/2} \cU_N (t;0) \xi \right\| \nn\\ &\leq& C \hbar e^{c |t|}  \| (\cN+1)^{k/2} \cU_N (t;0) \xi \|^2 
\eea
for a suitable constant $C >0$ (depending on $k$). Consider now the second term on the right hand side of (\ref{eq:erroest}),

\[\begin{split} \text{II} := \Big| \frac{1}{N} \int dx dy\, V (x-y) \, &\Big\langle \cU_N (t;0) \xi , (\cN+1)^{j-1} \\
&\times a (\cc v_{t,x}) a (\cc v_{t,y}) a (u_{t,y}) a(u_{t,x}) (\cN+1)^{k-j}  \cU_N (t;0) \xi \Big\rangle \Big|. \end{split}\]
Inserting a $1 = (\Ncal+5)^{k/2+1-j}(\Ncal+5)^{-k/2-1 + j}$ and pulling $(\Ncal+5)^{-k/2-1+j}$ through the annihilation operators to the right, we get:
\bea
\text{II} &\leq& \frac{1}{N} \int dp dx dy\, |\widehat{V} (p)| \left\|  (\cN+1)^{j-1} (\cN +5)^{k/2+1-j}\cU_N (t;0) \xi \right\| \nn\\ &&\quad \times \left\| \int dx dy \, a(\cc v_{t,x}) e^{ip\cdot x} a (u_{t,x}) a (\cc v_{t,y}) e^{-ip\cdot y} a (u_{t,y}) (\cN+1)^{k/2-1} \cU_N (t;0) \xi \right\|.\label{eq:Ibd}
\eea
Using that $v_{t}u_{t}=0$ and that $u_{t} = 1 - \o_{N,t}$, we obtain, for any $\phi\in\fock$:
\bea
\left\|\int \di x\, a(\cc v_{t,x}) e^{ipx} a(u_{t,x}) \phi \right\|  &=& \left\| \int \di r_1 \di r_2 \di x \, v_t(r_1;x) e^{ipx} u_t(x;r_2) a_{r_1} a_{r_2} \phi \right\| \nn\\ 
&\leq& \|v_{t}[e^{ip \cdot x},\o_{N,t}]\|_{\text{HS}}\| \cN^{1/2}\phi \|\nn\\
&\leq& 2^{1/2}\left\| [e^{ip\cdot x},\o_{N,t}] \right\|_{\text{tr}}^{1/2}\|\cN^{1/2}\phi\|\nn\\
&\leq& (2K\hbar (1+|p|) N)^{1/2} e^{c|t|} \, \| \cN^{1/2} \phi \| \label{eq:Ibd2}
\eea
where the second line follows from Lemma~\ref{lm:bds-2}, the third from $\|v_{t}\|\leq 1$, $\|e^{ip\cdot x}\|\leq 1$, $\|\o_{N,t}\|\leq 1$, while the last follows from Proposition \ref{lem:hbargain} (the constants $K$, $c > 0$ depend on $V$ but not on $k$). Applying this bound twice, we can estimate the last norm in the r.h.s. of (\ref{eq:Ibd}) as:
\bea
&&\Big\| \int dx dy \, a(\cc v_{t,x}) e^{ip\cdot x} a (u_{t,x}) a (\cc v_{t,y}) e^{-ip\cdot y} a (u_{t,y}) (\cN+1)^{k/2-1} \cU_N (t;0) \xi \Big\| \nn\\ &&\quad\leq (2K\hbar (1+|p|) N)^{1/2} e^{c|t|} \left\| \int dy\, 
a (\cc v_{t,y}) e^{-ip\cdot y} a (u_{t,y}) \cN^{1/2} (\cN+1)^{k/2-1}\cU_N (t;0) \xi \right\| \nn\\ 
&&\quad\leq 2K \hbar (1+|p|) N e^{2c|t|} \left\| \cN (\cN+1)^{k/2-1} \cU_N (t;0) \xi \right\| \nn\\ &&\quad\leq 2K\hbar (1+|p|) N e^{2c|t|} \left\| (\cN+1)^{k/2} \cU_N (t;0) \xi \right\|.
\eea
{P}lugging this bound into (\ref{eq:Ibd}), we conclude that
\bea
\text{II} &\leq& 2K \hbar  \left(\int dp \, |\widehat{V} (p)| (1+|p|) \right) e^{2 c |t|}  \left\| (\cN+5)^{k/2} \cU_N (t;0) \xi \right\|^2 \nn\\
&\leq& C \hbar e^{2 c |t|} \left\| (\cN + 1)^{k/2} \cU_N (t;0) \xi \right\|^2\nn
\eea
where the constant $c > 0$ depends on $V$ while the constant $C > 0$ depends on $V$ and on $k$. The last term in (\ref{eq:erroest}) is bounded analogously to $\text{I}$. This completes the proof of (\ref{eq:erroest}).
\end{proof}

\medskip

\begin{proof}[Proof of Theorem \ref{thm:grow}]
Combining Proposition \ref{prop:dt} and Lemma \ref{lem:errorest}, we find
\[ \left| i\hbar \frac{d}{dt} \, \Big\langle \cU_N (t;0) \xi, (\cN+1)^k \cU_N (t;0) \xi \Big\rangle \right| \leq C \hbar e^{c_1 |t|} \, \Big\langle  \cU_N (t;0) \xi, (\cN+1)^k \cU_N (t;0) \xi \Big\rangle. \]
Gronwall's Lemma implies that 
\[  \Big\langle \cU_N (t;0) \xi, (\cN+1)^k \, \cU_N (t;0) \xi \Big\rangle \leq \exp (c_2 \exp (c_1 |t|)) \,  \big\langle \xi, (\cN+1)^k \xi \big\rangle \]
where the constant $c_1$ depends only on the potential $V$, while $c_2$ depends on $V$ and on $k \in \bN$. 
\end{proof}

\section{Proof of main results}
\label{sec:proof}
\setcounter{equation}{0}

In this section we prove our main results, Theorem \ref{thm:main} and Theorem \ref{thm:k}. 
As in Section \ref{sec:gro}, we will use the notation $R_t \equiv R_{\nu_{N,t}}$, $u_t \equiv u_{N,t}$, $v_t \equiv v_{N,t}$, $\cc v_t \equiv \cc v_{N,t}$. Moreover, we define the functions $u_{t,x} (y) = u_{N,t} (y;x)$, $v_{t,x} (y) = v_{N,t} (y;x)$ and $\cc v_{t,x} (y) = \cc v_{N,t} (y;x)$. 
\begin{proof}[Proof of Theorem \ref{thm:main}]
We start from the expression
\begin{equation}\label{eq:gammant} \begin{split} \gamma_{N,t}^{(1)} (x;y) &= \langle \psi_{N,t} , a_y^* a_x \psi_{N,t} \rangle \\ &= \langle e^{-i \cH_N t/\hbar} R_0 \xi_N, a_y^* a_x e^{-i\cH_N t/ \hbar} R_0 \xi_N \rangle\\ &= \langle  \xi_N, R_0^* e^{i \cH_N t/\hbar} a_y^* a_x e^{-i\cH_N t/ \hbar} R_0 \xi_N \rangle. \end{split} \end{equation}
Introducing the fluctuation dynamics $\cU_N$ defined in (\ref{eq:cU}), we obtain
\[ \begin{split} 
\gamma_{N,t}^{(1)} (x;y) &=  \big\langle  \xi_N, \cU_N^* (t;0) R_{t}^* a_y^* a_x R_t \cU_N (t;0) \xi_N \big\rangle \\ & = \big\langle \xi_N , \cU_N^* (t;0) \left( a^* (u_{t,y}) + a ( \cc v_{t,y}) \right) \left( a (u_{t,x}) + a^* (\cc v_{t,x}) \right) \cU_N (t;0) \xi_N \big\rangle \\
 & = \big\langle \xi_N , \cU_N^* (t;0) \big\{ a^* (u_{t,y}) a (u_{t,x}) - a^* (\cc v_{t,x}) a (\cc v_{t,y}) + \langle \cc v_{t,y} , \cc v_{t,x} \rangle \\  &\hspace{3cm}  + a^* (u_{t,y}) a^* (\cc v_{t,x}) + a (\cc v_{t,y}) a (u_{t,x}) \big\} \, \cU_{N} (t;0) \xi_N \big\rangle. \end{split} \]
Here we used the defining property (\ref{eq:Bog}) of the Bogoliubov transformation $R_{t}$ and, in the third line, the canonical anticommutation relations (\ref{eq:CAR}). We observe that
\[ \langle \cc v_{t,y}, \cc v_{t,x} \rangle = \int dz v_t (z;y) \cc v_t (z;x) = (v_t \cc v_t) (y;x) = \omega_{N,t} (x;y). \]
This implies that
\[ \begin{split} 
\gamma_{N,t}^{(1)} (x;y) - \omega_{N,t} (x;y) = \; & \big\langle \xi_N, \cU_N^* (t;0) \big\{ a^* (u_{t,y}) a (u_{t,x}) - a^* (\cc v_{t,x}) a (\cc v_{t,y}) \\ &\hspace{1cm} + a^* (u_{t,y}) a^* (\cc v_{t,x}) + a (\cc v_{t,y}) a (u_{t,x}) \big\} \, \cU_N (t;0) \xi_N \big\rangle. \end{split} \]

\medskip

\emph{Step 1: Proof of \eqref{eq:conv-HS}.} We integrate this difference against the integral kernel of a Hilbert-Schmidt operator $O$ on $L^2(\Rbb^3)$ and find
\begin{equation} \label{eq:hs-tr} \begin{split} 
\tr \, O &\left( \gamma_{N,t}^{(1)} - \omega_{N,t} \right) \\ = \; & \big\langle \xi_N, \cU_N^* (t;0) \left(d\Gamma (u_t O u_t) -d\Gamma (\cc v_t \cc{O^{\ast}} v_t) \right)  \, \cU_N (t;0) \xi_N \big\rangle \\ &+ 2 \text{Re } \big\langle \xi_N, \cU_N^* (t;0) \left( \int dr_1 dr_2 (v_t O u_t) (r_1; r_2) a_{r_1} a_{r_2} \right) \cU_N (t;0) \xi_N \big\rangle.
 \end{split} \end{equation}
{F}rom Lemma \ref{lm:bds-2}, and using $\norm{u_t} = \norm{v_t} = 1$, we conclude that
 \begin{equation}\label{eq:hs-tr2}\begin{split} \left| \tr O \left( \gamma_{N,t}^{(1)} - \omega_{N,t} \right) \right| \leq \; &\left( \| u_t O u_t \| + \| \cc v_t \cc{O^{\ast}} v_t \| \right) \big\langle \xi_N, \cU_N^* (t;0) \cN \cU_N (t;0) \xi_N \big\rangle \\ &+ 2 \| v_t O u_t \|_{\text{HS}} \left\| (\cN+1)^{1/2} \cU_N (t;0) \xi_N \right\| \norm{\xi_N} \\  \leq \; & C \| O \|_{\text{HS}} \, \big\langle \xi_N, \cU_N^* (t;0) (\cN+1) \cU_N (t;0) \xi_N \big\rangle. \end{split}\end{equation}
{F}rom Theorem \ref{thm:grow} and from the assumption $\langle \xi_N, \cN \xi_N \rangle \leq C$, we obtain  
\[ \left\| \gamma_{N,t}^{(1)} - \omega_{N,t} \right\|_{\text{HS}} \leq C \exp (c_1 \exp (c_2 |t|)) \]
which completes the proof of (\ref{eq:conv-HS}).

\medskip

\emph{Step 2: Proof of \eqref{eq:conv-tr0}.}
We start anew from (\ref{eq:hs-tr}), assuming now $O$ to be a compact operator on $L^2(\Rbb^3)$, not necessarily Hilbert-Schmidt. Proceeding as in (\ref{eq:hs-tr2}), we find
\[ \begin{split} \Big| \tr\, & O \left( \gamma_{N,t}^{(1)} - \omega_{N,t} \right) \Big| \\ \leq \; &2 \| O \| \,  \| (\cN+1)^{1/2} \cU_N (t;0) \xi_N \|^2 + 2 \| v_t O u_t \|_{\text{HS}} \left\| (\cN+1)^{1/2} \cU_N (t;0) \xi_N \right\| \norm{\xi_N} \\  \leq \; & 2 \| O \| \, \| (\cN+1)^{1/2} \cU_N (t;0) \xi_N \|^2  + \| O \| \, \| v_t \|_{\text{HS}} \| (\cN+1)^{1/2} \cU_N (t;0) \xi_N \| \norm{\xi_N} .
\end{split}\]
Applying Theorem \ref{thm:grow}, the assumption $\langle \xi_N, \cN \xi_N \rangle \leq C$, and $\| v_t \|_{\text{HS}} = N^{1/2}$, we obtain
\[\Big| \tr O \left( \gamma_{N,t}^{(1)} - \omega_{N,t} \right) \Big| \leq C N^{1/2} \exp (c_1 \exp (c_2 |t|)) \,. \]
This completes the proof of \eqref{eq:conv-tr0}.

\medskip

\emph{Step 3: Proof of \eqref{eq:conv-tr}.} Let us now assume additionaly $d\Gamma (\omega_N) \xi_{N}=0$. Let $\xi^{(n)}_N$ the $n$-particle component of the Fock space vector $\xi_N$. With a slight abuse of notation, we denote again by $\xi_N^{(n)}$ the Fock space vector $\{ 0, \dots, 0 , \xi_N^{(n)}, 0 , \dots \} \in\cF$. The assumption implies that $d\Gamma (\omega_N) \xi_N^{(n)} = 0$ for all $n\in \bN$. Hence 
\be
\mathcal{N}R_{\n_N}\xi_N^{(n)} = R_{\n_N}(\mathcal{N} + N - 2d\G(\omega_N))\xi_N^{(n)} = R_{\n_N}(n + N)\xi_N^{(n)} = (n+N) R_{\nu_N} \xi_N^{(n)} \;. \label{eq:tr1}
\ee
In other words, $R_{\n_N}\xi_N^{(n)}$ is an eigenstate of the number of particles operator with eigenvalue $n+N$. Hence
\bea
&&\g_{N,t}^{(1)}(x;y) \nn \\ &&\quad= \sum_{n\geq 0}\media{e^{-i\mathcal{H}_{N}t/\hbar}R_{\n_N}\xi_N^{(n)},a^{*}_{y}a_{x}e^{-i\mathcal{H}_N t/\hbar}R_{\n_N}\xi_N^{(n)}}\nn\\
&&\quad= \sum_{n\geq 0}\media{\mathcal{U}_{N}(t;0)\xi_N^{(n)},(a^{*}(u_{t,y}) + a(\bar v_{t,y}))(a(u_{t,x}) + a^{*}(\bar v_{t,x}))\mathcal{U}_N(t;0)\xi_N^{(n)}}\nn.
\eea
Proceeding as in the proof of the first part of Theorem \ref{thm:main}, for a compact operator $O$ on $L^2(\Rbb^3)$, we end up with:
\begin{equation}\label{eq:tr2}
\begin{split}
\tr O &(\g^{(1)}_{N,t} - \omega_{N,t}) \\ =& \; \sum_{n\geq 0}\media{\xi_N^{(n)},\mathcal{U}^*_{N}(t;0)\left(d\G(u_t O u_t) - d\G(\cc v_t \cc{O^{\ast}} v_t)\right)\mathcal{U}_N(t;0)\xi_N^{(n)}} \\
&+ 2\Re \sum_{n\geq 0} \media{\xi_N^{(n)},\cU^*_N(t;0)\left( \int dr_1 dr_2 (v_t O u_t) (r_1; r_2) a_{r_1} a_{r_2} \right) \cU_N (t;0) \xi_N^{(n)}} \\
=: & \; \text{I} + \text{II}\;.
\end{split} \end{equation}
We estimate separately the two lines in the r.h.s. of Eq. (\ref{eq:tr2}). Let us start with
\be
\text{I} = \sum_{n\geq 0}\media{\xi_N^{(n)},\mathcal{U}^*_{N}(t;0)\left(d\G(u_t O u_t) - d\G(\cc v_t \cc{O^{\ast}} v_t)\right)\mathcal{U}_N(t;0)\xi_N^{(n)}}\nn.
\ee
From Lemma \ref{lm:bds-2}, we get
\begin{equation}\label{eq:Ifin_}
\begin{split}
\text{I} &\leq (\| u_t O u_t \| + \| \cc v_t \cc{O^{\ast}} v_t \|)\sum_{n\geq 0}\media{\xi_N^{(n)},\cU^*_N(t;0)\mathcal{N}\cU_N(t;0)\xi_N^{(n)}} \\
&\leq C \|O\| \exp(c_1\exp(c_2|t|)) \sum_{n\geq 0}  \media{\xi_N^{(n)},\mathcal{N}\xi_N^{(n)}}\\
&= C\|O\| \exp(c_1\exp(c_2|t|)) \media{\xi_N,\mathcal{N}\xi_N} \leq C\|O\| \exp(c_1\exp(c_2|t|)) 
\end{split}
\end{equation}
where we used the fact that $\|u_{t}\| = \|v_{t}\| = 1$, Theorem \ref{thm:grow} to control the growth of the expectation of $\mathcal{N}$ w.r.t. the fluctuation dynamics $\cU_N$, and that by assumption $\langle \xi_N , \cN \xi_N \rangle \leq C$. Therefore, we are left with
\be
\text{II} = 2\Re \sum_{n\geq 0} \media{\xi_N^{(n)},\cU^*_N(t;0)\left( \int dr_1 dr_2 (v_t O u_t) (r_1; r_2) a_{r_1} a_{r_2} \right) \cU_N (t;0) \xi_N^{(n)}}\;.\label{eq:tr3}
\ee
It follows from Proposition \ref{prop:dt} that, for any $k \in \bN$, 
\[\begin{split} i\hbar &\frac{d}{dt} \left\langle \cU_N^* (t;0) \xi, \cN^k \cU_N (t;0) \xi \right \rangle = \left\langle \xi, \cU_N^* (t;0) [\cN^k , \cL_N (t) ] \, \cU_N (t;0) \xi \right\rangle \end{split}  \] 
where we defined  
\begin{equation}\label{eq:genLN} \begin{split} 
\cL_N (t) = \; &\frac{1}{N} \int dx dy\, V(x-y) \big\{ a^*(u_{t,x}) a(\cc v_{t,y}) a(u_{t,y}) a(u_{t,x}) \\ &\hspace{2.5cm} +  \frac{1}{2}a(\cc v_{t,x}) a(\cc v_{t,y}) a(u_{t,y}) a(u_{t,x}) - a^*(u_{t,y}) a^*(\cc v_{t,y}) a^*(\cc v_{t,x}) a(\cc v_{t,x}) \big\}\\&+\hc
\end{split} \end{equation}
The operator $\cL_N (t)$ plays the role of the generator of the dynamics $\cU_N (t;0)$ (up to terms which commute with the number of particles operator). Approximating the orthogonal projection $P_n = { \bf 1} (\cN = n)$ on the $n$-particle sector of $\cF$ by polynomials in $\cN$, we conclude that 
\begin{equation}\label{eq:fcNder}   i\hbar \frac{d}{dt} \left\langle \cU_N^* (t;0) \xi, f (\cN) \, \cU_N (t;0) \xi \right \rangle = \left\langle \xi, \cU_N^* (t;0) [ f(\cN) , \cL_N (t) ]  \,\cU_N (t;0) \xi \right\rangle \end{equation}
for every continuous function $f: \bR \to \bC$. We are going to compare the dynamics $\cU_N (t;0)$ with a modified dynamics, whose generator only contains one of the three terms on the r.h.s. of (\ref{eq:genLN}) (and terms which commute with $\cN$, therefore not contributing to the change of the expectation of functions of $\cN$). We define
\begin{equation}\label{eq:cL2}
\begin{split} \wt{\cL}_N (t) = \frac{1}{N} & \int dxdy V(x-y) \\ &\times \left\{ a^*(u_{t,x}) a(\cc v_{t,y}) a(u_{t,y}) a(u_{t,x}) - a^*(u_{t,y}) a^*(\cc v_{t,y}) a^*(\cc v_{t,x}) a(\cc v_{t,x}) \right\} + \text{h.c.} \end{split} \end{equation} 
and we denote by $\wt{\cU}_N (t;s)$ the time evolution generated by $\cU^*_N (t;0) \wt{\cL}_N (t) \cU_N (t;0)$ which is the two-parameter group of unitary transformation satisfying $\wt{\cU}_N (s;s) = 1$ for all $s \in \bR$ and 
\begin{equation}\label{eq:Utilde} i \hbar \, \partial_t \, \wt{\cU}_N (t;s) = -\cU_N^* (t;0) \wt{\cL}_N (t) \,\cU_N (t;0) \, \wt{\cU}_N (t;s). \end{equation}
Finally, we define 
\[ \cU_N^{(1)} (t;s) = \cU_N (t;s) \, \wt{\cU}_N (t;s) \]
and we observe that, from (\ref{eq:fcNder}) and (\ref{eq:Utilde}),
\begin{equation}\label{eq:fcN} 
\begin{split}
i \hbar \partial_t \langle \xi, \cU^{(1)*}_N (t;0) f(\cN) &\, \cU^{(1)}_{N} (t;0) \xi \rangle = \left\langle \xi, \cU^{(1)*}_N (t;0) \left[ f (\cN) , \cL^{(1)}_N (t) \right] \cU^{(1)}_N (t;0) \xi \right\rangle
\end{split} 
\end{equation}
where we set 
\[ \begin{split} \cL_N^{(1)} (t) &= \cL_N (t) - \wt{\cL}_N (t) \\
&= \frac{1}{2N} \int dx dy V(x-y) \left\{a(\cc v_{t,x}) a(\cc v_{t,y}) a(u_{t,y}) a(u_{t,x}) + a^* (u_{t,x}) a^* (u_{t,y}) a^* (\cc v_{t,y}) a^* (\cc v_{t,y})  \right\} \end{split} \]
Notice that $\cL^{(1)}_N (t)$ can only create or annihilate four particles at a time. This implies that, although $\cL_N^{(1)} (t)$ does not commute with $\cN$, it satisfies $[\cL^{(1)}_N(t) , i^\cN] = 0$. 
{F}rom (\ref{eq:fcN}), we conclude that
\[ \langle \xi, \cU^{(1)*}_N (t;0) i^\cN \cU^{(1)}_{N} (t;0) \xi \rangle =  \langle \xi, i^\cN \, \xi \rangle \]
for all $t \in \bR$ and all $\xi \in \cF$. We conclude that $\cU^{(1)*}_N (t;0) \, i^\cN \, \cU_N^{(1)} (t;0) = i^{\cN}$ for all $t \in \bR$ and therefore that  \begin{equation}\label{eq:iNcomm} i^\cN \cU^{(1)}_N (t;0) = \cU^{(1)}_N (t;0) \, i^\cN \, . \end{equation}

\medskip

We rewrite now (\ref{eq:tr3}) as 
\begin{equation} \label{eq:tr5}
\begin{split}
\text{II} =  & \; 2\Re \sum_{n\geq 0}\Big\{\media{\xi_N^{(n)}, \cU^{(1)*}_N(t;0)\left( \int dr_1 dr_2 (v_t O u_t) (r_1; r_2) a_{r_1} a_{r_2} \right) \cU^{(1)}_N(t;0)\xi_N^{(n)}} \\
& + \media{\xi_N^{(n)},\left( \cU^{*}_{N}(t;0) - \cU^{(1)*}_N(t;0) \right)\left( \int dr_1 dr_2 (v_t O u_t) (r_1; r_2) a_{r_1} a_{r_2} \right)\cU^{(1)}_{N}(t;0)\xi_N^{(n)}} \\
& +\media{\xi_N^{(n)},\cU^{*}_N(t;0)\left( \int dr_1 dr_2 (v_t O u_t) (r_1; r_2) a_{r_1} a_{r_2} \right)\left( \cU_N(t;0) - \cU^{(1)}_N(t;0) \right)\xi_N^{(n)}}\Big\} \\
=: & \;\text{II}_{1} + \text{II}_{2} + \text{II}_{3}\;.
\end{split}
\end{equation}
The key remark which allows us to improve the rate of convergence with respect to Eq. (\ref{eq:conv-tr0}) is that $\text{II}_1 = 0$. This follows from the remark that $\cU^{(1)}_N$ can only create or annihilate particles in groups of four. So, the expectation of the a product of two creation or two annihilation operators in the state $\cU_N^{(1)} (t;0) \xi_N^{(n)}$, where $\xi_N^{(n)}$ has a fixed number of particles must vanish. To prove this fact rigorously, we use (\ref{eq:iNcomm}), which implies that
\bea
&&\media{\xi_N^{(n)}, \cU^{(1)*}_N(t;0)\left( \int dr_1 dr_2 (v_t O u_t) (r_1; r_2) a_{r_1} a_{r_2} \right) \cU^{(1)}_N(t;0)\xi_N^{(n)}} \nn\\
&&\quad = \media{\xi_N^{(n)}, \cU^{(1)*}_N(t;0)\left( \int dr_1 dr_2 (v_t O u_t) (r_1; r_2) a_{r_1} a_{r_2} \right) \cU^{(1)}_N(t;0)i^{\cN}\xi_N^{(n)}}i^{-n}\nn\\
&&\quad = \media{\xi_N^{(n)}, i^{\cN + 2}\cU^{(1)*}_N(t;0)\left( \int dr_1 dr_2 (v_t O u_t) (r_1; r_2) a_{r_1} a_{r_2} \right) \cU^{(1)}_N(t;0)\xi_N^{(n)}}i^{-n} \nn\\
&&\quad = - \media{\xi_N^{(n)}, \cU^{(1)*}_N(t;0)\left( \int dr_1 dr_2 (v_t O u_t) (r_1; r_2) a_{r_1} a_{r_2} \right) \cU^{(1)}_N(t;0)\xi_N^{(n)}} \nn,
\eea
and therefore 
\[\media{\xi_N^{(n)}, \cU^{(1)*}_N(t;0)\left( \int dr_1 dr_2 (v_t O u_t) (r_1; r_2) a_{r_1} a_{r_2} \right) \cU^{(1)}_N(t;0)\xi_N^{(n)}} \nn =0.\]
We are left with bounding the last two terms in (\ref{eq:tr5}); let us start with 
\be
\begin{split}
\text{II}_{2} = \; &2\Re \sum_{n\geq 0} \Big \langle \xi_N^{(n)},\left( \cU^{*}_{N}(t;0) - \cU^{(1)*}_N(t;0) \right)\\ &\hspace{3cm} \times \left( \int dr_1 dr_2 (v_t O u_t) (r_1, r_2) a_{r_1} a_{r_2} \right)\cU^{(1)}_{N}(t;0)\xi_N^{(n)} \Big\rangle.\label{eq:tr6}
\end{split} 
\ee
We expand $\cU_N$ in terms of $\cU^{(1)}_N$ using the Duhamel formula:
\bea
\cU_N(t;0) - \cU^{(1)}_N(t;0) = -\frac{i}{\hbar}\int_{0}^{t}ds\,\cU(t;s)\wt{\cL}_N (s) \cU^{(1)}_N(s;0). \label{eq:tr7}
\eea
Plugging (\ref{eq:tr7}) into (\ref{eq:tr6}) and using (\ref{eq:cL2}) we end up with
\begin{equation}
\begin{split}
\text{II}_{2} \leq& \; \frac{4}{\hbar N}\sum_{n\geq 0}\left\{\left| \left\langle\xi_N^{(n)}, \int_{0}^{t} ds \, \cU^{(1)*}_{N}(s;0) \right.\right.\right.\\ &\hspace{2cm} \times  \left(\int dxdy\, V(x-y) a^{*}(u_{s,x})a(\bar v_{s,y}) a(u_{s,y}) a(u_{s,x}) + \text{h.c.} \right) \\ &\hspace{2cm} \times \left.\left.\left.\cU^{*}_N (t;s) \left(\int dr_1 dr_2 (v_t O u_t) (r_1; r_2) a_{r_1} a_{r_2}\right) \cU^{(1)}_{N}(t;0) 
\xi_N^{(n)} \right\rangle \right|\right. \\&
 \left.+ \left|\left\langle \xi_N^{(n)}, \int_{0}^{t}ds\, \cU^{(1)*}_N(s;0) \left( \int dxdy\,V(x-y) a^{*}(u_{s,y})a^{*}(\bar v_{s,y})a^{*}(\bar v_{s,x})a(\bar v_{s,x}) + \text{h.c.} \right)\right.\right.\right. \\ &\hspace{2cm} \times\left.\left.\left. \cU^{*}_N(t;s) \left( \int dr_1 dr_2 (v_t O u_t) (r_1; r_2) a_{r_1} a_{r_2} \right) \cU^{(1)}(t;0)\xi_N^{(n)}\right\rangle\right|\right\} \\
=:& \;\text{II}_{2.1} + \text{II}_{2.2}\,.\label{eq:tr8}
\end{split}
\end{equation}
We start by estimating $\text{II}_{2.1}$. We find
\[
\begin{split}
\text{II}_{2.1} \leq& \; \frac{2}{\hbar N}\sum_{n\geq 0}\int_{0}^{t}ds \int dp\,|\hat V(p)| \\ & \times \Big\{\Big| \Big\langle \xi_N^{(n)}, \cU^{(1)*}_{N}(s;0)d\G(u_s e^{ipx}\bar u_s) \\& \hspace{1cm} \times\left.\left.\left( \int d\o_1 d\o_2\, (v_s e^{-ipx}\bar u_s)(\o_1;\o_2)a_{\o_1}a_{\o_2} \right)\cU^{*}_{N}(t;s)\right.\right. \\& \hspace{1cm}\times\left( \int dr_1 dr_2 (v_t O u_t) (r_1; r_2) a_{r_1} a_{r_2} \right)\cU^{(1)}_{N}(t;0) \xi_N^{(n)}\Big\rangle \Big| \\
&\hspace{0.3cm} +\Big| \Big \langle \xi_N^{(n)}, \cU^{(1)*}_{N}(s;0)\left( \int d\o_1 d\o_2\, (\bar v_s e^{-ipx} u_s)(\o_1;\o_2)a^*_{\o_1}a^*_{\o_2} \right) \\& \hspace{1cm}\times d\G(u_s e^{ipx}\bar u_s )\cU^{*}_{N}(t;s)\left( \int dr_1 dr_2 (v_t O u_t) (r_1; r_2) a_{r_1} a_{r_2} \right)\cU^{(1)}_{N}(t;0) \xi_N^{(n)}\Big\rangle \Big|\Big\}\;.
\end{split}\]
Using Cauchy-Schwarz inequality, we obtain 
\begin{equation}
\begin{split}
\text{II}_{2.1} \leq& \; \frac{2}{\hbar N}\sum_{n\geq 0}\int_{0}^{t}ds\int dp|\hat V(p)|\left\| d\G(u_s e^{ipx} \bar u_s)\cU^{(1)}_{N}(s;0)\xi_N^{(n)} \right\| \\&\hspace{1cm} \times \left\| \left( \int d\o_1 d\o_2\, (v_s e^{-ipx}\bar u_s)(\o_1;\o_2)a_{\o_1}a_{\o_2}\right)\cU^{*}_{N}(t;s)\right. \\ &\hspace{1.5cm} \times \left.\left(\int dr_1 dr_2 (v_t O u_t) (r_1; r_2) a_{r_1} a_{r_2}\right)\cU^{(1)}_N(t;0)\xi_N^{(n)} \right\| \\
&+\frac{2}{\hbar N}\sum_{n\geq 0}\int_{0}^{t}ds\int dp\, |\hat V(p)| \\ &\hspace{1cm} \times \left\| (\cN + 2)^{-1/2}d\G(u_s e^{ipx}\bar u_s)\right.\\&\hspace{1.5cm} \left. \times \left( \int d\o_1 d\o_2\, (v_s e^{-ipx} \bar u_s)(\o_1;\o_2)a_{\o_1}a_{\o_2}\right)\cU^{(1)}_{N}(s;0)\xi_N^{(n)}\right\| \\& \hspace{1cm} \times \left\| (\cN + 2)^{1/2}\cU^{*}_{N}(t;s)\left( \int dr_1 dr_2 (v_t O u_t) (r_1; r_2) a_{r_1} a_{r_2} \right)\cU^{(1)}_{N}(t;0)\xi_N^{(n)}\right\|.
\end{split} \end{equation}
{F}rom Lemma \ref{lm:bds-2}, it follows that 
\begin{equation}
\begin{split}
\text{II}_{2.1} \leq & \; \frac{2}{\hbar N}\sum_{n\geq 0}\int_{0}^{t}ds\int dp\,|\hat V(p)|\left\| \cN \cU^{(1)}_N(s;0) \xi_N^{(n)}\right\| \left\| v_s e^{-ipx}\bar u_s \right\|_{\text{HS}} \\ &\hspace{1cm} \times \left\| \cN^{1/2}\cU^{*}_{N}(t;s)\left( \int dr_1 dr_2 (v_t O u_t) (r_1; r_2) a_{r_1} a_{r_2} \right)\cU^{(1)}_{N}(t;0)\xi_N^{(n)}\right\| \\
&+ \frac{2}{\hbar N}\sum_{n\geq 0}\int_{0}^{t}ds\int dp\,|\hat V(p)|\| v_{s}e^{-ipx}\bar u_{s} \|_{\text{HS}}\left\| \cN \cU^{(1)}_{N}(s;0)\xi_N^{(n)} \right\| \\
&\hspace{1cm} \times\left\| (\cN + 2)^{1/2}\cU^{*}_N(t;s)\left( \int dr_1 dr_2 (v_t O u_t) (r_1; r_2) a_{r_1} a_{r_2} \right)\cU^{(1)}_{N}(t;0)\xi_N^{(n)}\right\|\,.
\end{split} \end{equation} 
Using Theorem \ref{thm:grow} to control the growth of $\cN$ with respect to the unitary evolutions, and again Lemma \ref{lm:bds-2}, we conclude that
\begin{equation} \label{eq:finII2}
\begin{split}
\text{II}_{2.1} \leq \; &\frac{C \exp (c_1 \exp (c_2 |t|))}{\hbar N} \\ &\times \sum_{n\geq 0}\int_{0}^{t}ds\int dp\,|\hat V(p)| \left\| v_s e^{-ipx}\bar u_s \right\|_{\text{HS}}\left\|v_t O u_t\right\|_{\text{HS}} \left\| (\cN + 2) \xi_N^{(n)}\right\|^2. \end{split}
\end{equation}
Here we also used a bound of the form $\| \cN \cU_N^{(1)} (t;0) \xi \| \leq C \exp (c_1 \exp (c_2 |t|)) \| \cN \xi \|$ for the growth of the expectation of the number of particles w.r.t. to the modified dynamics $\cU^{(1)}_N (t;0)$. This bound can be proven exactly as the estimate in Theorem \ref{thm:grow} for the dynamics $\cU_N (t;0)$, with the only difference that when we compute the derivative of $\langle \xi , \cU^{(1)}_N (t;0) (\cN+1)^k \cU_N^{(1)} \xi \rangle$ only one of the three terms on the r.h.s. of (\ref{eq:prop-dt}) appears. 

\medskip

Since $\left\|u_{t} O v_t\right\|_{\text{HS}}\leq \|O\|N^{1/2}$ and, using Proposition \ref{lem:hbargain}, 
\[ \left\| v_{s} e^{ipx} \bar u_s \right\|_{\text{HS}}^{2} \leq \Tr \left|[\g_{s},e^{ipx}]\right| \leq C (1+|p|) N \hbar  \exp (c |s|), \] we find that
\[ 
\begin{split}
\text{II}_{2.1} &\leq C \| O \| \hbar^{-1/2} \exp (c_1 \exp (c_2 |t|)) \sum_{n\geq 0} \| (\cN+2) \xi_N^{(n)} \|^2 \\ &\leq C \| O \| \hbar^{-1/2} \exp (c_1 \exp (c_2 |t|)) \| (\cN+2) \xi_N \|^2.
\end{split} \]
The same strategy is followed to bound $\text{II}_{2.2}$ in (\ref{eq:tr8}), and $\text{II}_{3}$ in (\ref{eq:tr5}). Hence, we have shown that, for every compact operator $O$, 
\[  \begin{split} \left| \tr O \, (\g^{(1)}_{N,t} - \omega_{N,t})\right| &\leq C \|O\| N^{1/6} \exp{(c_2\exp(c_1|t|))} \media{\xi_N,(\cN+2)^{2}\xi_N} \\ &\leq C \|O\| N^{1/6} \exp{(c_2\exp(c_1|t|))} \end{split} \]
where we used the assumption $\langle \xi_N , \cN^2 \xi_N \rangle \leq C$. This completes the proof of (\ref{eq:conv-tr}). 

\medskip

\emph{Step 4: Proof of (\ref{eq:conv-sc}).} We consider an observable $O = e^{ix\cdot q + \hbar p \cdot \nabla}$ with $p,\,q\in\mathbb{R}^{3}$. As in (\ref{eq:tr2}) we decompose 
\[\tr O (\g^{(1)}_{N,t} - \omega_{N,t}) = \; \text{I} + \text{II}. \] 
The bound for $\text{I}$ obtained in (\ref{eq:Ifin_}) for an arbitrary bounded operator $O$ is already consistent with (\ref{eq:conv-sc}). However, we have to improve the bound for $\text{II}$, using the special structure of the observable $O$. Writing $\text{II} = \text{II}_1 + \text{II}_2 + \text{II}_3$ as in (\ref{eq:tr5}), and noticing again that $\text{II}_1 = 0$, we are left with the problem of improving the bound for $\text{II}_2$ and $\text{II}_3$. To bound $\text{II}_2$, we use (\ref{eq:finII2}) and the remark that, for $O= e^{ix\cdot q + \hbar \nabla\cdot p}$,
\begin{equation}\label{eq:uJv}
\begin{split}
\left\| v_{t} O u_{t} \right\|_{\text{HS}}^{2} &= \left\| v_{t}e^{ix\cdot q + \hbar p \cdot \nabla}u_{t} \right\|_{\text{HS}}^{2} \\ &\leq \tr\left| \left[\omega_{N,t},e^{ix\cdot q + \hbar p \cdot \nabla}\right] \right| 
\leq \tr\left| \left[\omega_{N,t},e^{ix\cdot q}\right] \right| + \tr\left|\left[\omega_{N,t},e^{\hbar p\cdot \nabla}\right]\right|.
\end{split} \end{equation}
Using that
\[ \begin{split} [ \omega_{N,t} , e^{\hbar p \cdot \nabla} ] &= \omega_{N,t} e^{\hbar p \cdot \nabla} - e^{\hbar p \cdot \nabla} \omega_{N,t} = - \int_0^1 ds \, \frac{d}{ds} e^{s \hbar p \cdot \nabla} \omega_{N,t} e^{(1-s) \hbar p \cdot \nabla} \\ &= - \int_0^1 ds \, e^{s \hbar p \cdot \nabla} [\hbar p \cdot \nabla , \omega_{N,t}] e^{(1-s) \hbar p \cdot \nabla} \end{split} \]
we conclude from Proposition \ref{lem:hbargain} that
\[ \tr |[\omega_{N,t} , e^{\hbar p \cdot \nabla}]| \leq |p| \tr |[ \hbar \nabla , \omega_{N,t}]| \leq C |p| N \hbar \exp (c |t|). \]

\medskip

Therefore, using Proposition \ref{lem:hbargain} also to bound $\tr | [\omega_{N,t},e^{ix\cdot q}]|$,   (\ref{eq:uJv}) implies that
\[ \| v_{t} O u_{t} \|_{\text{HS}}^{2}  \leq C (1+|q| + |p|) N \hbar \exp (c |t|). \]
Inserting this bound in (\ref{eq:finII2}), we obtain that, for $O= e^{ix\cdot q + \hbar \nabla\cdot p}$,
\[ \text{II}_2 \leq C \exp (c_1 \exp (c_2 |t|)) (1+|p| + |q|)^{1/2} \| (\cN+1) \xi_N \|^2. \]
A similar bound can be found for the contribution $\text{II}_3$. Hence 
\[ \begin{split} \left| \tr \, e^{ix\cdot q + \hbar \nabla\cdot p} \left( \gamma_{N,t}^{(1)} - \omega_{N,t} \right) \right| & \leq C (1+|p| + |q|)^{1/2} \exp (c_1 \exp (c_2 |t|))  \| (\cN+1) \xi_N \|^2 \\ 
 & \leq C (1+|p| + |q|)^{1/2} \exp (c_1 \exp (c_2 |t|)),
\end{split} \]
where we used the assumption $\| (\cN+1) \xi_N \|^2 < C$. This concludes the proof of Theorem~ \ref{thm:main}.
\end{proof}

Next, we proceed with the proof of Theorem \ref{thm:k}. 
 \begin{proof}[Proof of Theorem \ref{thm:k}] 
We start from the expression
\begin{equation}\label{eq:k1} \begin{split} \gamma_{N,t}^{(k)} &(x_1, \dots , x_k ; x'_1, \dots x'_k) \\ &= \big\langle e^{-i \cH_N t/ \hbar} 
R_0 \xi,  a^*_{x'_k}  \dots a^*_{x'_1} a_{x_1} \dots a_{x_k} e^{-i\cH_N t/\hbar} R_0 \xi \big\rangle \\ &=  \big\langle \cU_N (t;0) \xi, R_t^* a^*_{x'_k}  \dots a^*_{x'_1} a_{x_1} \dots a_{x_k} R_t \cU_N (t;0)  \xi \big\rangle \\
&=  \big\langle \cU_N (t;0) \xi, \left(a^*(u_{t,x'_k})+a(\cc v_{t,x'_k}) \right)\cdots\left(a^*(u_{t,x'_1})+a(\cc v_{t, x'_1}) \right) \\
  & \qquad\qquad\times \left(a(u_{t,x_1})+a^*(\cc v_{t,x_1}) \right)\cdots\left( a(u_{t,x_k})+a^*(\cc v_{t,x_k}) \right) \, \cU_N (t;0) \xi \big\rangle.
  \end{split}
 \end{equation}
 This product will be expanded as a sum of $2^{2k}$ summands. Each summand will be put in normal order using Wick's theorem, which gives rise to contractions. The completely contracted contribution will be identified with the Hartree-Fock density matrix $\omega^{(k)}_{N,t}$, all other contributions will be of smaller order.

\medskip

\emph{Step 1: Expanding the product and applying Wick's theorem.} We recall Wick's theorem. For $j =1 ,\dots , 2k$, we denote by $a^\#_j$ either an annihilation or a creation operator acting on the fermionic Fock space $\cF$. We denote by $:a_{j_1}^\# \dots a_{j_\ell}^\#:$ the product $a_{j_1}^\# \dots a_{j_\ell}^\#$ in normal order (obtained by moving all creation operators on the left and all annihilation operators on the right, proceeding as if they were all anticommuting operators). Wick's theorem states that
\[ \begin{split} 
a^\#_1 a^\#_2 \cdots a^\#_{2k} = \; : a^\#_1 a^\#_2 \cdots a^\#_{2k} : + &\sum_{j=1}^k \sum_{n_1 < \dots < n_{2j}} :a^\#_1 \cdots \widehat{a^\#_{n_1}}\cdots \widehat{a^\#_{n_{2j}}}\cdots a^\#_{2k}: \\ &\times \sum_{\sigma \in P_{2j}} (-1)^{|\sigma|} \scal{\Omega}{a^\#_{n_{\sigma(1)}} a^\#_{n_{\sigma (2)}}\Omega}\cdots\scal{\Omega}{a^\#_{n_{\sigma (2j-1)}} a^\#_{n_{\sigma (2j)}}\Omega}
\end{split} \]
where $P_{2j}$ is the set of pairings 
\[
\begin{split}P_{2j} = \big\{ \sigma \in S_{2j}:\ & \sigma(2\ell-1)< \sigma(2\ell)\ \forall \ell=1,\ldots, j \text{ and } \\ & \sigma(2\ell-1) < \sigma(2\ell+1)\ \forall \ell=1,\ldots,j-1\big\}.\end{split}\]
and $|\sigma|$ denotes the number of pair interchanges needed to bring the contracted operators in the order $a_{n_{\sigma (1)}}^\# a_{n_{\sigma (2)}}^\# \dots a_{n_{\sigma (2j)}}^\#$. We call $\scal{\Omega}{a^\#_i a^\#_j\Omega}$ the contraction of $a^\#_i$ and $a^\#_j$. 

\medskip

Next, we apply Wick's theorem to the products arising from (\ref{eq:k1}). To this end, we observe that  
the contraction of a $a^\#(u_{t,z_1})$-operator with a $a^\#(\cc v_{t,z_2})$-operator is always zero because $u_t \cc v_t = v_t u_t =0$. Furthermore, the $a^\#(u_{t,z})$-operators among themselves are already in normal order, so their contractions always vanish. Hence, the only non-vanishing contractions arising from the terms on the r.h.s. of (\ref{eq:k1}) have the form 
\begin{equation}\label{gammacontr}\scal{\Omega}{a(\cc v_{t,x'_i}) a^*(\cc v_{t,x_j}) \Omega} = \omega_{N,t}(x_j,x'_i) \, . \end{equation} 
Since each contraction of the form (\ref{gammacontr}) involves one $x$- and one $x'$-variable, the normal-ordered products in the non-vanishing contributions arising from Wick's theorem always have the same number of $x$- and $x'$-variables. So, all terms emerging from (\ref{eq:k1}) after applying Wick's theorem have the form
\begin{equation}
\label{eq:typicalwick}
\begin{split}
& \pm \Big\langle \cU_N (t;0) \xi, \, :a^\#(w_1(\cdot;x'_{\sigma(1)}))\cdots a^\#(w_{k-j}(\cdot;x'_{\sigma(k-j)})) \\ &\hspace{2.5cm} \times a^\#(\eta_1(\cdot;x_{\pi(1)}))\cdots a^\#(\eta_{k-j}(\cdot;x_{\pi(k-j)})): \, \cU_N (t;0) \xi \Big\rangle \\
& \hspace{5cm} \times\omega_{N,t} (x_{\pi(k-j+1)};x'_{\sigma(k-j+1)})\cdots\omega_{N,t} (x_{\pi(k)};x'_{\sigma(k)})
\end{split}
\end{equation}
where $j \leq k$ denotes the number of contractions, $\pi, \sigma \in S_k$ are two appropriate permutations, and, for every $j=1, \dots , k-j$, $w_j, \eta_j: L^2 (\bR^3) \to L^2 (\bR^3)$ are either the  operator $u_t$ or the operator $\cc v_t$ (the operators are identified with their integral kernels). 

\medskip

\emph{Step 2: Estimating \eqref{eq:typicalwick} in the case $0 \leq j < k$.} We will use the shorthand notation $\bx_k = (x_1, \dots ,x_k) \in \bR^{3k}$ and similarly $\bx'_k = (x'_1, \dots , x'_k) \in \bR^{3k}$. Let $O$ be a Hilbert-Schmidt operator on $L^2 (\bR^{3k})$, with integral kernel $O(\bx_k ; \bx'_k)$. Integrating \eqref{eq:typicalwick} against $O(\bx_k ; \bx'_k)$, we set 
\begin{equation}\label{eq:wickbound}
\begin{split}
\text{I} & := \Big|  \int d \bx_k d\bx'_k \, O (\bx_k;\bx'_k) \, \big\langle \cU_N (t;0) \xi, \, :a^\#(w_1(\cdot;x'_{\sigma(1)}))\cdots a^\#(w_{k-j}(\cdot;x'_{\sigma(k-j)})) \\
&\hspace{2.5cm} \times a^\#(\eta_1(\cdot; x_{\pi(1)}))\cdots a^\#(\eta_{k-j}(\cdot; x_{\pi(k-j)})): \cU_N (t;0)\xi \big\rangle \\
& \hspace{5cm} \times\omega_{N,t} (x_{\pi(k-j+1)};x'_{\sigma(k-j+1)})\cdots\omega_{N,t} (x_{\pi(k)};x'_{\sigma(k)})\Big|. 
\end{split} \end{equation}
We remark that 
\[ \begin{split} 
\text{I} & = \Big| \int d\bx_k d \bx'_k \left[ \eta^{(\pi (1))}_1\cdots \eta^{(\pi (k-j))}_{k-j}  \, O \, w^{(\sigma (1))}_1 \cdots w^{(\sigma (k-j))}_{k-j} \right]  (\bx_k ; \bx'_k) \\
& \hspace{2.5cm} \times \big\langle \cU_N (t;0) \xi, \, :a^\#_{x'_{\sigma(1)}} \cdots a^\#_{x'_{\sigma(k-j)}} a^\#_{x_{\pi(1)}} \cdots a^\#_{x_{\pi(k-j)}}: \, \cU_N (t;0) \xi \big\rangle\\
&\hspace{5cm} \times\omega_{N,t} (x_{\pi(k-j+1)}; x'_{\sigma(k-j+1)})\cdots \omega_{N,t} (x_{\pi(k)}; x'_{\sigma(k)})\Big| 
\end{split}\]
where $\eta_\ell^{(\pi (\ell))}$ and $w_\ell^{(\sigma (\ell))}$ denote the one-particle operators $\eta_\ell$ and $w_\ell$ acting only on particle $\pi (\ell)$ and, respectively, on particle $\sigma (\ell)$.  Notice that to be precise some of the operators $\eta_\ell^{(\pi(\ell))}$ and $w_\ell^{(\sigma (\ell))}$ may need to be replaced by their transpose, their complex conjugate, or their hermitian conjugate. In the end, this change does not affect our analysis, since we will only need the bounds $\| \eta_j \|, \| w_j \| \leq 1$ for the operator norms.  {F}rom H\"older's inequality, we get
\[
\begin{split} 
\text{I} \leq \; & \big\|  \eta^{(\pi (1))}_1\cdots \eta^{(\pi (k-j))}_{k-j}  \, O \, w^{(\sigma (1))}_1 \cdots w^{(\sigma (k-j))}_{k-j} \big\|_{\text{HS}} \\ &\times  \bigg( \int d\bx_k d \bx'_k \, \big| \big\langle \cU_N (t;0) \xi, \, :a^\#_{x'_{\sigma(1)}}\cdots a^\#_{x'_{\sigma(k-j)}}a^\#_{x_{\pi(1)}}\cdots a^\#_{x_{\pi(k-j)}}: \cU_N(t;0)\xi\big\rangle \big|^2 \\ & \hspace{2cm} \times  \big| \omega_{N,t} (x_{\pi(k-j+1)};x'_{\sigma(k-j+1)}) \big|^2 \cdots \big|\omega_{N,t} (x_{\pi(k)},;x'_{\sigma(k)}) \big|^2 \bigg)^{1/2}\\
 \leq \; &  \| O \|_{\text{HS}} \| \omega_{N,t} \|_{\text{HS}}^j  \\
&\times \bigg( \int \di x_{\pi(1)}\cdots \di x_{\pi(k-j)}\di x'_{\sigma(1)}\cdots dx'_{\sigma(k-j)} \\
&\hspace{2cm} \times  \big| \big\langle \cU_N (t;0) \xi, :a^\#_{x'_{\sigma(1)}}\cdots a^\#_{x'_{\sigma(k-j)}}a^\#_{x_{\pi(1)}}\cdots a^\#_{x_{\pi(k-j)}}: \cU_N (t;0) \xi \big\rangle \big|^2 \bigg)^{1/2}.
\end{split}\]
Since $\| \omega_{N,t} \|_{\text{HS}} = N^{1/2}$ and since the operators in the inner product are normal ordered, we obtain 
\[ \text{I} \leq C \| O \|_{\text{HS}} N^{j/2} \langle \cU_N (t;0) \xi , (\cN+1)^{k-j} \cU_N (t;0) \xi \rangle \,.\]
Hence, the contribution of each term with $j < k$ arising from (\ref{eq:k1}) after applying Wick's theorem and integrating against a Hilbert-Schmidt operator $O$ can be bounded by 
\begin{equation}\label{eq:jmink} C \| O \|_{\text{HS}} \, N^{(k-1)/2} \langle \cU_N (t;0) \xi , (\cN+1)^{k} \cU_N (t;0) \xi \rangle \,. \end{equation}

\medskip

{\it Step 3: Fully contracted terms, $j = k$}. To finish the proof of Theorem \ref{thm:k}, we consider the fully contracted terms with $j=k$ arising from (\ref{eq:k1}) after expanding and applying Wick's theorem. Since $\scal{\Omega}{a(\cc v_{t,\cdot,y_i}) a^*(\cc v_{t,x_j})\Omega} = \omega_{N,t}(x_j;y_i)$ are the only non-zero contractions, only the term 
\[ a (\cc v_{t,y_k}) \cdots a(\cc v_{t, y_1}) a^*(\cc v_{t,x_1}) \cdots a^*(\cc v_{t,x_k})\]
on the r.h.s. of (\ref{eq:k1}) produces a non-vanishing, fully contracted, contribution. {F}rom (\ref{gammacontr}) and comparing with the definition (\ref{eq:omk}), this contribution is given by 
\[ \sum_{\pi \in S_{k}} \sgn(\pi) \,  \omega_{N,t} (x_{1}; x'_{\pi (1)}) \dots \omega_{N,t} (x_{k}; x'_{\pi (k)}) = \omega_{N,t}^{(k)} (\bx_k ; \bx'_k). \] 

\medskip

Combining the results of Step 2 and Step 3, we conclude that 
\[ \left| \tr \, O \left( \gamma_{N,t}^{(k)} - \omega_{N,t}^{(k)} \right) \right| \leq C N^{(k-1)/2} \| O \|_{\text{HS}} \langle \cU_N (t;0) \xi , (\cN+1)^{k-1} \cU_N (t;0) \xi \rangle \]
for every Hilbert-Schmidt operator $O$ on $L^2(\Rbb^{3k})$. Eq. (\ref{eq:k-claim}) now follows from Theorem \ref{thm:grow}. 

\medskip

{\it Step 4: Bound for the trace norm.} Eq. (\ref{eq:k-claim-tr}) follows, similarly to (\ref{eq:k-claim}), if we can show that, for any bounded operator $O$ on $L^2 (\bR^{3k})$, the contribution (\ref{eq:wickbound}) can be bounded by 
\begin{equation}\label{eq:term-tr} \text{I} \leq C \| O \| \, N^{\frac{k+j}{2}} \exp (c_1 \exp (c_2 |t|)) \end{equation}
for all $\xi \in \cF$ with $\langle \xi, \cN^{k} \xi \rangle < \infty$, and the number of contractions $0 \leq j < k$. In fact, because of the fermionic symmetry of $\gamma_{N,t}^{(k)}$ and $\omega_{N,t}^{(k)}$, it is enough to establish (\ref{eq:term-tr}) for all bounded $O$ with the symmetry
\[ O (x_{\pi (1)}, \dots , x_{\pi (k)} ; x'_{\sigma (1)}, \dots , x'_{\sigma (k)}) = \text{sgn} (\pi) \text{sgn} (\sigma) \, O (x_1, \dots , x_k; x'_1, \dots x'_k) \]
for any permutations $\pi, \sigma \in S_k$.
For such observables, (\ref{eq:wickbound}) can be rewritten as
\[ \begin{split}  \text{I} & = \Big|  \int d \bx_k d\bx'_k \, O (\bx_k,\bx'_k) \, \big\langle \cU_N (t;0) \xi, \, :a^\#(w_1(\cdot,x'_1))\cdots a^\#(w_{k-j} (\cdot,x'_{k-j})) \\
&\hspace{2.5cm} \times a^\#(\eta_1(\cdot, x_{1}))\cdots a^\#(\eta_{k-j}(\cdot, x_{k-j})): \cU_N (t;0)\xi \big\rangle \\
& \hspace{5cm} \times\omega_{N,t} (x_{k-j+1},x'_{k-j+1})\cdots\omega_{N,t} (x_{k},x'_{k})\Big| \\
& = \Big| \int d\bx_{k-j} d \bx'_{k-j} \left[ \eta^{(1)}_1\cdots \eta^{(k-j)}_{k-j}  \,\left(  \tr_{k-j+1, \dots , k} \, O  (1 \otimes \omega_{N,t}^{\otimes j}) \right) \, w^{(1)}_1 \cdots w^{(k-j)}_{k-j} \right]  (\bx_{k-j} ; \bx'_{k-j}) \\ & \hspace{2.5cm} \times \big\langle \cU_N (t;0) \xi, \, :a^\#_{x'_1} \cdots a^\#_{x'_{k-j}} a^\#_{x_{1}} \cdots a^\#_{x_{k-j}}: \, \cU_N (t;0) \xi \big\rangle
\end{split}\]
where 
\[\begin{split}  \Big( \tr_{k-j+1, \dots , k}  O  &(1 \otimes \omega_{N,t}^{\otimes j}) \Big) (\bx_{k-j} ; \bx'_{k-j}) \\ & = \int dx_{k-j+1} dx'_{k-j+1} \dots dx_k dx'_k \, O (\bx_k ; \bx'_k) \prod_{\ell=k-j+1}^k \omega_{N,t} (x_\ell; x'_\ell) \end{split} \]
denotes the partial trace over the last $j$ particles. Using Cauchy-Schwarz, we obtain
\bea
\text{I} &\leq& \left\| \eta^{(1)}_{1} \dots \eta^{(k-j)}_{k-j} \, \left(\tr_{k-j+1, \dots , k} \, O \, (1 \otimes \o_{N,t}^{\otimes j}) \right) w^{(1)}_{1} \dots w^{(k-j)}_{k-j} \right\|_{\text{HS}} \, \left\| \cN^{\frac{k-j}{2}} \cU_N(t;0)\xi \right\|^{2}\nn\\
&\leq&  \left\| \eta^{(1)}_{1} \dots \eta^{(k-j)}_{k-j} \right\|_{\text{HS}} \, \left\| \tr_{k-j+1, \dots , k}  \, O (1 \otimes \o_{N,t}^{\otimes j}) \right\| \, \left\| \cN^{\frac{k-j}{2}}\cU_N(t;0)\xi \right\|^{2}\nn\\
&\leq& N^{\frac{k-j}{2}} \, \left\| \tr_{k-j+1, \dots , k} \, O (1 \otimes \o_{N,t}^{\otimes j}) \right\| \, \left\| \cN^{\frac{k-j}{2}}\cU_N(t;0)\xi \right\|^{2}\;\nn
\eea
where in the second line we used that $\|w^{(j)}_{j}\|=1$ for all $j=1, \dots , k-j$. Since
\bea
\left\| \tr_{k-j+1, \dots , k} \, O (1\otimes \o_{N,t}^{\otimes j}) \right\| &=& \sup_{{\substack{\phi, \ph \in L^{2}(\mathbb{R}^{3(k-j)}) \\ \|\phi\| = \| \psi \| \leq 1}}} \left| \left\langle \phi, \left(\tr_{k-j+1, \dots , k} \, O (1 \otimes \o_{N,t}^{\otimes j}) \right) \varphi\right\rangle \right| \nn\\
&=& \sup_{{\substack{\phi, \ph \in L^{2}(\mathbb{R}^{3(k-j)}) \\ \|\phi\| = \| \psi \| \leq 1}}} 
\left| \tr \, O \left(|\ph \rangle \langle \phi| \otimes \o^{\otimes j}_{N,t}\right) \right| \nn\\
&\leq& \left(\tr \lvert \omega_{N,t} \rvert\right)^j \norm{O} \leq N^j \, \|O\| \;,\label{eq:trk2}
\eea
we get
\be
\text{I} \leq N^{\frac{k+j}{2}} \, \|O\| \, \left\| \cN^{\frac{k-j}{2}}\cU_N(t;0)\xi \right\|^{2}\,,\nn
\ee
which, by Theorem \ref{thm:grow}, proves (\ref{eq:term-tr}).
\end{proof}

\section{Propagation of semiclassical structure}
\label{sec:sc}
\setcounter{equation}{0}

In this section we prove Proposition \ref{lem:hbargain}, which propagates the bounds (\ref{eq:sc}) along the solution of the Hartree-Fock equation and plays a central role in our analysis. 
\begin{proof}[Proof of Proposition \ref{lem:hbargain}] 
Let $\omega_{N,t}$ denote the solution of the Hartree-Fock equation (\ref{eq:hf}). We define the (time-dependent) Hartree-Fock Hamiltonian
\[ h_{\text{HF}} (t)  = -\hbar^2 \Delta + (V * \rho_t) (x)- X_t \]
where $\rho_t (x) = (1/N) \omega_{N,t} (x;x)$ and $X_t$ is the exchange operator, having the kernel $X_t (x;y) = V(x-y) \omega_{N,t} (x;y)$. Then $\omega_{N,t}$ satisfies the equation
\[ i\hbar\partial_t \omega_{N,t} = [ h_{\text{HF}} (t) , \omega_{N,t} ]. \]
Therefore, we obtain
\begin{equation}\label{eq:HFcomm}
\begin{split}
i\hbar \frac{d}{dt} [e^{ip\cdot x},\omega_{N,t}] &= [e^{ip \cdot x},[h_{\text{HF}} (t) , \omega_{N,t}]] \\ 
&= [h_{\text{HF}} (t),[e^{ip\cdot x}, \omega_{N,t}]] + [ \omega_{N,t}, [h_{\text{HF}} (t), e^{ip \cdot x}]] \\
&= [h_{\text{HF}} (t),[e^{ip\cdot x}, \omega_{N,t}]] - [ \omega_{N,t}, [\hbar^2 \Delta, e^{ip \cdot x}]] - [ \omega_{N,t}, [ X_t , e^{ip\cdot x}]] 
\end{split}
\end{equation} 
where we used the cyclic properties of the commutator and the fact that $[\r_{t}*V,e^{ipx}] = 0$.
We compute
\[ [\hbar^2 \Delta, e^{ip \cdot x}] = i \hbar \nabla \cdot \hbar p e^{ip \cdot x} + e^{ip \cdot x} \hbar p \cdot i\hbar \nabla \]
and hence
\[ \begin{split} [ \omega_{N,t} , [ \hbar^2 \Delta, e^{ip\cdot x} ]] &= [ \omega_{N,t},  i \hbar \nabla \cdot \hbar p e^{ip \cdot x} + e^{ip \cdot x} \hbar p \cdot i\hbar \nabla ] \\ &= [\omega_{N,t} , i\hbar \nabla] \cdot \hbar p e^{ip \cdot x} + i \hbar^2 \nabla \cdot p [ \omega_{N,t}, e^{ip \cdot x}]  \\ &\quad + \hbar p e^{ip \cdot x} [ \omega_{N,t}, i\hbar \nabla ] + [ \omega_{N,t} , e^{ip \cdot x} ] i \hbar^2 \nabla \cdot p\,.  \end{split} \]
{F}rom (\ref{eq:HFcomm}) we find
\[ \begin{split} i\hbar \frac{d}{dt} [ e^{ip\cdot x} , \omega_{N,t}] = &\; A(t) [ e^{ip\cdot x}, \omega_{N,t}] - [ e^{ip \cdot x}, \omega_{N,t} ] B(t) \\& - \hbar p e^{ip \cdot x} [ \omega_{N,t}, i\hbar \nabla ] - [\omega_{N,t} , i\hbar \nabla] \cdot \hbar p e^{ip \cdot x} -  [\omega_{N,t}, [ X_t , e^{ip\cdot x}]] \end{split} \]
where we defined the time-dependent operators
\[ A(t) = h_{\text{HF}} (t) + i \hbar^2 \nabla \cdot p \quad \text{and } \quad B (t) = h_{\text{HF}} (t) - i\hbar^2 \nabla \cdot p\,.\]
Observe that $A(t)$ and $B(t)$ are self-adjoint for every $t \in \bR$ (the factor $\pm i \hbar^2 p \cdot \nabla$ can be interpreted as originating from a constant vector potential).  They generate two unitary evolutions $\cU_1 (t;s)$ and $\cU_2 (t;s)$ satisfying 
\[ i \hbar \partial_t \cU_1 (t;s) = A(t) \cU_1 (t;s) \quad \text{ and } \quad i \hbar \partial_t \cU_2 (t;s) = B(t) \cU_2 (t;s) \]
with the initial conditions $\cU_1 (s;s) = \cU_2 (s;s) = 1$. We observe that, by definition of the unitary maps $\cU_1 (t;s)$ and $\cU_2 (t;s)$, 
\[ \begin{split} 
i \hbar \frac{d}{dt} \, \cU^*_1 (t;0) &[e^{ip\cdot x}, \omega_{N,t} ] \, \cU_2 (t;0) \\ = \; & \cU^*_1 (t;0) \left\{ - A(t) [e^{ip\cdot x}, \omega_{N,t} ] + [ e^{ip \cdot x} , \omega_{N,t}] B(t)  + i\hbar \frac{d}{dt} [ e^{ip \cdot x}, \omega_{N,t} ] \right\} \cU_2 (t;0) \\ = \; & -\cU^*_1 (t;0) \big\{  \hbar p e^{ip \cdot x} [ \omega_{N,t}, i\hbar \nabla ] + [\omega_{N,t} , i\hbar \nabla] \cdot \hbar p e^{ip \cdot x} + [\omega_{N,t}, [ X_t , e^{ip\cdot x}]] \big\} \, \cU_2 (t;0). \end{split} \]
Hence
\[ \begin{split} \cU^*_1 &(t;0) [ e^{ip\cdot x}, \omega_{N,t}] \, \cU_2 (t;0) \\ = \; & [e^{ip \cdot x}, \omega_N] \\ &+ \frac{i}{\hbar} \int_0^t ds \, \cU_1^* (s;0) \big\{  \hbar p e^{ip \cdot x} [ \omega_{N,s}, i\hbar \nabla ] + [\omega_{N,s} , i\hbar \nabla] \cdot \hbar p e^{ip \cdot x} + [\omega_{N,s}, [ X_s , e^{ip\cdot x}]] \big\} \cU_2 (s;0) \end{split} \]
and therefore
\[ \begin{split} [e^{ip\cdot x}&, \omega_{N,t}] \\ = \; &\cU_1 (t;0) [ e^{ip\cdot x}, \omega_N ] \, \cU_2^* (t;0) \\ &+ \frac{i}{\hbar} \int_0^t ds \, \cU_1 (t;s) \big\{  \hbar p e^{ip \cdot x} [ \omega_{N,s} , i\hbar \nabla ] + [\omega_{N,s} , i\hbar \nabla] \cdot \hbar p e^{ip \cdot x} + [\omega_{N,s}, [ X_s , e^{ip\cdot x}]] \big\} \cU_2 (s;t). \end{split} \]
Taking the trace norm, we find
\begin{equation}\label{eq:est-comm1} \tr | [e^{ip\cdot x}, \omega_{N,t}] | \leq \tr | [e^{ip\cdot x}, \omega_N] | + 2|p| \int_0^t ds \, \tr | [\hbar \nabla , \omega_{N,s}] | + \frac{1}{\hbar} \int_0^t ds \tr |[\omega_{N,s},  [ X_s , e^{ip\cdot x}]]|. \end{equation}
To control the contribution of the last term, we observe that 
\[ X_s (x;y) = \frac{1}{N} \, V(x-y) \omega_{N,s} (x;y) = \frac{1}{N} \int dq \, \widehat{V} (q) e^{iq \cdot (x-y)} \omega_{N,s} (x;y) = \frac{1}{N} \int dq \, \widehat{V} (q) \omega_{q,t} (x;y) \]
where we defined the operator $\omega_{q,t} = e^{iq \cdot x} \omega_{N,t} e^{-i q \cdot x}$ (here $x$ indicates the multiplication operator). Hence, we get
\[ [ \omega_{N,t}, [X_t, e^{ip\cdot x}]] = \frac{1}{N} \int dq \, \widehat{V} (q) [ \omega_{N,t} , [ \omega_{q,t} , e^{i p \cdot x}] ] \]
and therefore, using $\norm{\omega_{N,t}} \leq 1$,
\[ \begin{split} \tr |[ \omega_{N,t}, [X_t, e^{ip\cdot x}]]| &\leq \frac{1}{N} \int dq |\widehat{V} (q)| \, \tr | [ \omega_{N,t} , [ \omega_{q,t} , e^{i p \cdot x}] ]| \\ & \leq \frac{2}{N} \int dq |\widehat{V} (q)| \, \tr |[ \omega_{q,t} , e^{i p \cdot x}]| \\ & \leq \frac{2}{N} \left( \int dq |\widehat{V} (q)| \right) \, \tr |[ \omega_{N,t} , e^{ip\cdot x}]| 
\end{split} \]
where in the last line we used that $[\omega_{q,t},e^{ipx}] = e^{iqx}[\omega_{N,t},e^{ipx}]e^{-iqx}$. {F}rom (\ref{eq:est-comm1}), we conclude that
\begin{equation}\label{eq:est-comm2} \tr | [e^{ip\cdot x}, \omega_{N,t}] | \leq \tr  | [e^{ip\cdot x}, \omega_N] | + 2|p| \int_0^t ds \, \tr | [\hbar \nabla , \omega_{N,s}] | + \frac{C}{N\hbar} \int_0^t ds \tr \, | [e^{ip\cdot x}, \omega_{N,s}] | 
\end{equation}
and therefore, from (\ref{h1}), we find
\begin{equation}\label{eq:est-comm1-fin}
\begin{split}
\sup_p \frac{1}{1+|p|} \tr | [e^{ip\cdot x}, \omega_{N,t}] | \leq \; &C \hbar N + 2 \int_0^t ds \, \tr \, | [\hbar \nabla , \omega_{N,s}] | \\ &+  C \int_0^t ds \, \sup_p \frac{1}{1+|p|} \tr\, |[e^{ip\cdot x}, \omega_{N,s}] | .\end{split} \end{equation}

\medskip

Next, we need to control the growth of $\tr|[\hbar\nabla, \omega_{N,t}]|$. Consider
\[ \begin{split}
i\hbar \frac{d}{dt}[\hbar\nabla , \omega_{N,t} ] =& [\hbar \nabla, [h_{\text{HF}} (t) , \omega_{N,t} ]] \\
=& [h_{\text{HF}} (t) ,[\hbar\nabla, \omega_{N,t}]] + [\omega_{N,t}, [ h_{\text{HF}} (t), \hbar \nabla]] \\ =& [h_{\text{HF}} (t) ,[\hbar\nabla, \omega_{N,t}]] + [\omega_{N,t}, [ V*\rho_t , \hbar \nabla]] - [ \omega_{N,t}, [X_t, \hbar \nabla]] \, .
\end{split} \]
As before, the first term on the r.h.s. can be eliminated by an appropriate unitary conjugation. Denote namely by $\cU_3 (t;s)$ the two-parameter unitary group satisfying
\[ i\hbar \partial_t \, \cU_3 (t;s) = h_{\text{HF}} (t) \, \cU_3 (t;s) \]
and $\cU_3 (s;s) = 1$. Then we compute
\[ \begin{split} i\hbar \frac{d}{dt} \cU_3^* (t;0) [ \hbar \nabla , \omega_{N,t}] \,  \cU_3 (t;0) &= \cU_3^* (t;0) \left\{ - [ h_{\text{HF}} (t) , [\hbar\nabla, \omega_{N,t}]] + i\hbar \frac{d}{dt}  [ \hbar \nabla , \omega_{N,t}] \right\} \cU_3 (t;0) 
\\ &= \cU_3^* (t;0) \left\{  [\omega_{N,t}, [ V*\rho_t , \hbar \nabla]] - [ \omega_{N,t}, [X_t, \hbar \nabla]]  \right\} \cU_3 (t;0). \end{split} \]
This gives
\[ \begin{split} [ \hbar \nabla , \omega_{N,t}]  = \; & \cU_3 (t;0) [ \hbar \nabla , \omega_N] \cU^*_3 (t;0) \\ &+ \frac{1}{i\hbar} \int_0^t ds \, \cU_3 (t;s) \left\{ [\omega_{N,s}, [ V*\rho_s , \hbar \nabla]] - [ \omega_{N,s}, [X_s, \hbar \nabla]]  \right\} \cU_3 (s;t) \end{split} \]
and therefore
\begin{equation}\label{eq:est-comm3} \begin{split} 
\tr | [ \hbar \nabla , \omega_{N,t}] | & \leq \tr \, | [ \hbar \nabla , \omega_N] |\\
&\quad + \frac{1}{\hbar} \int_0^t ds \, \tr | [\omega_{N,s} ,  [ V*\rho_s , \hbar \nabla]] | + \frac{1}{\hbar} \int_0^t ds \, \tr \, | [ \omega_{N,s}, [X_s, \hbar \nabla]]|. \end{split} \end{equation}
The second term on the r.h.s. can be controlled by
\[ \begin{split} 
\tr | [\omega_{N,s} ,  [ V*\rho_s , \hbar \nabla]] | &= \hbar \, \tr |[\omega_{N,s} , \nabla (V * \rho_s)]| \\ &\leq \hbar \int dq \, |\widehat{V} (q)||q| |\widehat{\rho}_s (q)|  \, \tr |[ \omega_{N,s} , e^{iq \cdot x} ] | \\ &\leq \hbar \left( \int dq \, |\widehat{V} (q)| (1+|q|)^{2} \right)  \sup_q \frac{1}{1+|q|} \tr  |[ \omega_{N,s} , e^{iq \cdot x} ] | \end{split} \]
where we used the bound $\| \widehat{\rho}_s \|_\infty \leq \| \rho_s \|_1 = 1$. 
As for the last term on the r.h.s. of (\ref{eq:est-comm3}), we note that
\[ [\omega_{N,s}, [ X_s, \hbar\nabla]] = \frac{1}{N}\int dq\, \widehat{V} (q) [\omega_{N,s} , [\omega_{q,s}, \hbar \nabla]] \]
where, as above, we set $\omega_{q,s} = e^{iq \cdot x} \omega_{N,s} e^{-iq\cdot x}$.  Hence, we obtain
\[ \begin{split} \tr | [\omega_{N,s}, [ X_s, \hbar\nabla]]| &\leq \frac{2}{N} \int dq \, |\widehat{V} (q)| \tr |[\omega_{q,s}, \hbar \nabla]| \\ &\leq \frac{2}{N} \left(\int dq \, |\widehat{V} (q)| \right) \,  \tr |[\omega_{N,s}, \hbar \nabla]| \end{split} \]
where in the last inequality we used that 
\[
[\omega_{q,s} , \hbar\nabla] = e^{iqx}[ \omega_{N,s} , \hbar(\nabla + iq)]e^{-iqx} = e^{iqx}[\omega_{N,s}, \hbar \nabla] e^{-iqx}\;.
\]
{F}rom (\ref{eq:est-comm3}), we conclude that
\[ 
\tr | [ \hbar \nabla , \omega_{N,t}] | \leq C \hbar N +  C \int_0^t ds \, \sup_q \frac{1}{1+|q|} \tr  |[ \omega_{N,s} , e^{iq \cdot x} ] | + C \int_0^t ds \, \tr |[\omega_{N,s}, \hbar \nabla]|.  \]
Summing up the last equation with (\ref{eq:est-comm1-fin}) and applying Gronwall's lemma, we 
find constants $C,c > 0$ such that 
\[ \begin{split} 
\sup_p \frac{1}{1+|p|} \tr | [e^{ip\cdot x}, \omega_{N,t}] | &\leq C \hbar N \exp (c |t|), \\ 
\tr | [ \hbar \nabla , \omega_{N,t}]  &\leq C \hbar N \exp (c |t|). \qedhere\end{split}\]
\end{proof}

\appendix

\section{Comparison between Hartree and Hartree-Fock dynamics}
\label{app:hartree}
\setcounter{equation}{0}

In the next proposition we show that, under the assumptions of Theorem \ref{thm:main}, the solution $\omega_{N,t}$ of the Hartree-Fock equation (\ref{eq:hf}) is well-approximated by the solution $\wt{\omega}_{N,t}$ of the Hartree equation (\ref{eq:h}). Since we can show that the difference $\omega_{N,t} - \wt{\omega}_{N,t}$ remains of order one in $N$ for all fixed times $t \in \bR$, this result implies that all bounds in Theorem \ref{thm:main} remain true if we replace $\omega_{N,t}$ by $\wt{\omega}_{N,t}$.
\begin{prp}\label{prop:hartree}
Assume that the interaction potential $V \in L^1 (\bR^3)$ satisfies (\ref{eq:ass-V}) and that the sequence $\omega_N$ of orthogonal projections on $L^2 (\bR^3)$ with $\tr \omega_N = N$ satisfies (\ref{eq:sc}). Let $\omega_{N,t}$ denote the solution of the Hartree-Fock equation 
\[ i\hbar \partial_t \omega_{N,t} = \left[ -\hbar^2 \Delta + (V* \rho_t) -X_t , \omega_{N,t} \right] \]
and $\wt{\omega}_{N,t}$ the solution of the Hartree equation
\[  i\hbar \partial_t \wt{\omega}_{N,t} = \left[ -\hbar^2 \Delta + (V* \wt{\rho}_t) ,\wt{\omega}_{N,t} \right] \]
with initial data $\omega_{N,t=0} = \wt{\omega}_{N,t=0} = \omega_N$ (recall here that $\rho_t (x) = N^{-1} \omega_{N,t} (x;x)$, $\wt{\rho}_t (x) = N^{-1} \wt{\omega}_{N,t} (x;x)$ and $X_t (x;y) = N^{-1} 
V(x-y)  \omega_{N,t} (x;y)$).
Then there exist constants $C,c_1, c_2 > 0$ such that 
\[ \tr \left| \omega_{N,t} - \wt{\omega}_{N,t} \right| \leq C \exp (c_1 \exp (c_2 |t|)) \]
for all $t \in \bR$. 
\end{prp}

\begin{proof}
Let $\cW (t;s)$ be the unitary dynamics generated by the Hartree Hamiltonian $h_{H} (t) = -\hbar^2 \Delta + (V* \wt{\rho}_t)$. In other words, $\cW (s;s) = 1$ for all $s \in \bR$ and 
\[
i\hbar\frac{d}{dt} \cW (t;s) = h_H (t) \cW (t;s) \,.
\]
Then, we have 
\[
\begin{split}
i \hbar\partial_t \cW^{*}(t;0) \widetilde\o_{N,t} \cW (t;0) &= 0\;, \\
i\hbar\partial_t \cW^{*} (t;0) \o_{N,t} \cW (t;0) &= \cW^{*} (t;0 )\left( \left[  V* \left( \r_t - \widetilde\r_{t}  \right), \o_{N,t} \right] - \left[ X_{t}, \o_{t} \right] \right) \cW (t;0). 
\end{split}
\]
Integrating over time, we end up with
\[
\begin{split}
\widetilde\o_{N,t} &= \cW (t;0) \o_{N}\cW^{*}(t;0)\, , \\
\o_{N,t} &= \cW (t;0) \o_{N} \cW^{*} (t;0) 
\\  & \quad - \frac{i}{\hbar}\int_{0}^{t} ds\,\cW (t;s) \left( \left[ V*\left( \widetilde\r_{s} - \r_{s} \right), \o_{N,s} \right] - \left[ X_{s},\o_{N,s} \right]\right)\cW^* (t;s) 
\end{split}
\]
and thus 
\begin{equation}\label{eq:HHF2c}
\begin{split}
\tr\left|\o_{N,t} - \widetilde\o_{N,t}\right| \leq & \; \frac{1}{\hbar}\int_{0}^{t} ds \left\{ \tr\left|\left[ V * \left( \widetilde\r_{s} - \r_{s} \right), \o_{N,s} \right]\right| + \tr\left|\left[ X_{s},\o_{N,s} \right]\right| \right\}
=: \; \text{I} + \text{II}.
\end{split}
\end{equation}
Let us first estimate $\text{II}$. We get
\bea\text{II} &=& \frac{1}{\hbar}\int_{0}^{t}ds\,\tr\left| \left[ X_{s},\o_{N,s} \right] \right| \nn\\
&\leq& \frac{1}{\hbar N}\int_{0}^{t}ds\int dp\,|\widehat V(p)|\tr\left|\left[ e^{ip\cdot x}\o_{N,s}e^{-ip\cdot x},\o_{N,s} \right|\right] \nn\\ 
&\leq& \frac{2}{\hbar N}\int_{0}^{t}ds\int dp\, |\widehat V(p)|\tr\left|\left[ e^{ip\cdot x},\o_{N,s} \right]\right|  \nn\\ 
&\leq& C \exp{(c|t|)},\label{eq:HHF3}
\eea
where in the last step we used Proposition \ref{lem:hbargain} ($e^{ip\cdot x}$ denotes here the multiplication operator). We are left with $\text{I}$. Writing
\[ V* (\wt{\rho}_s - \rho_s) (x) = \int dp \, \widehat{V} (p) \left(\widehat{\wt{\rho}}_s (p) - \widehat{\rho}_s (p) \right) \, e^{i p \cdot x} \]
we find
\bea
\text{I} &\leq& \frac{1}{\hbar}\int_{0}^{t}ds\int dp\, |\hat V(p)| \, \left|\widehat{\wt{\rho}}_s (p) - \widehat{\rho}_s (p)\right| \, \tr\left| \left[  e^{ip\cdot x}, \o_{N,s} \right]\right|\nn\\
&\leq& CN\exp(c|t|) \int_{0}^{t} ds\, \sup_{p \in\mathbb{R}^{3}} \left| \widehat{\r}_{s}(p) - \widehat{\widetilde\r}_{s}(p) \right|\nn\\
&\leq& C\exp(c|t|) \int_{0}^{t}ds\,\tr \left| \o_{N,s} - \widetilde\o_{N,s} \right| \label{eq:HHF3b}
\eea
where in the second inequality we used again Proposition \ref{lem:hbargain}, while in the last inequality we used the bound 
\be
\left|\widehat{ \r}_{s}(p) - \widehat{\widetilde{\r}}_{s}(p) \right| = \frac{1}{N}\left| \tr e^{ip\cdot x}(\o_{N,s} - \widetilde\o_{N,s}) \right| \leq \frac{1}{N} \tr| \o_{N,s} - \widetilde\o_{N,s}|.\nn
\ee
Inserting (\ref{eq:HHF3}), (\ref{eq:HHF3b}) into (\ref{eq:HHF2c}), and applying Gronwall lemma, we get
\be
\tr\left| \o_{N,t} - \widetilde\o_{N,t} \right| \leq C\exp{(c_1 \exp{(c_2 |t|)})}
\ee
for some $C$, $c_{1}$, $c_{2}$ only depending on the potential $V$.
\end{proof}

\end{document}